\documentclass{lmcs} %%% last changed 2014-08-20
\pdfoutput=1

\usepackage{lastpage}
\lmcsdoi{19}{4}{27}
\lmcsheading{}{\pageref{LastPage}}{}{}%
{Mar.~15,~2022}{Dec.~18,~2023}{}


\usepackage{hyperref}
\theoremstyle{plain} %\crefname{satz}{Satz}{S\"atze}
\usepackage[utf8]{inputenc}
\usepackage[english]{babel}
\usepackage{mathtools}
\usepackage[cal=cm]{mathalfa}
\usepackage{amssymb}
\usepackage{amsthm, thmtools}

\usepackage{multicol}

\usepackage{pgf, tikz}
\usetikzlibrary{automata,arrows, arrows.meta, positioning}
\usepackage{color}

\usepackage[T1]{fontenc}
\usepackage{microtype} 
\usepackage{bbold} %for mathbb 1
\usepackage{multirow} %for tabular
\usepackage{xspace}
\usepackage{setspace}

\definecolor{grau}{rgb}{0.47,0.47,0.47}
\definecolor{grau2}{rgb}{0.1,0.1,0.1}

\tikzset{
main edge/.style={line width=0.7pt, draw=grau2,>={Latex[width=3pt]},-to, },
main node/.style={state,fill=white,draw=grau2, line width=0.6pt, scale=0.8},
edge label/.style={midway,text=black, scale=1,draw=none}
}

\newcommand{\wafa}{\textup{WAFA}\xspace}
\newcommand{\wfta}{\textup{WFTA}\xspace}
\newcommand{\wfa}{\textup{WFA}\xspace}
\newcommand{\pola}{\textup{PA}\xspace}
\newcommand{\wafaa}{(Q,\Sigma,\delta,P_0,\tau)}
\newcommand{\wftaa}{(Q,\Gamma,\delta,\lambda)}
\newcommand{\N}{\mathbb{N}}
\newcommand{\A}{\mathcal{A}}
\newcommand{\B}{\mathcal{B}}

\newcommand{\pols}[1]{S\left[ #1 \right]}

\newcommand{\chr}[1]{\mathbb{1}_{#1}}
\newcommand{\bhv}[1]{[ \mkern-3mu [ #1 ] \mkern-3mu ]}
\newcommand{\bhvs}[1]{[ #1 ]}
\newcommand{\ser}[2]{#1 \langle \mkern-4mu \langle #2 \rangle \mkern-4mu \rangle}
\newcommand{\con}[1]{%
{\mathop{#1}\limits^{\vbox to -0.5\ex@{\kern-\tw@\ex@
   \hbox{\scriptsize $\sim$}\vss}}}}

\newcommand{\tr}{t_{\mathcal{R}}}
\newcommand{\tw}{t_w^r}
\newcommand{\tlg}{T_\Gamma}

\DeclareMathOperator{\pos}{Pos}
\DeclareMathOperator{\ra}{r}
\DeclareMathOperator{\pot}{\mathcal{P}}
\DeclareMathOperator{\rank}{Rank}
\DeclareMathOperator{\weight}{Weight}
\newcommand{\lbl}[1]{\operatorname{Label}_{#1}}

\newcommand{\oset}[2]{%
  {\mathop{#2}\limits^{\vbox to 1\ex@{\kern-\tw@\ex@
   \hbox to 14\ex@{\scriptsize $#1$}\vss}}}}
\newcommand{\cons}[1]{%
  {\mathop{#1}\limits^{\vbox to -1\ex@{\kern-\tw@\ex@
   \hbox{\scriptsize $\sim$}\vss}}}}
\def\moverlay{\mathpalette\mov@rlay}
\def\mov@rlay#1#2{\leavevmode\vtop{%
   \baselineskip\z@skip \lineskiplimit-\maxdimen
   \ialign{\hfil$\m@th#1##$\hfil\cr#2\crcr}}}
\newcommand{\charfusion}[3][\mathord]{
    #1{\ifx#1\mathop\vphantom{#2}\fi
        \mathpalette\mov@rlay{#2\cr#3}
      }
    \ifx#1\mathop\expandafter\displaylimits\fi}
\makeatother

\let\emptyset\varnothing%leereMeneg

\newcommand{\overbar}[1]{\mkern 1.5mu\overline{\mkern-1.5mu#1\mkern-1.5mu}\mkern 1.5mu}

\begin{document}

\title[A Nivat theorem for weighted alternating automata]{A Nivat theorem for weighted alternating automata over commutative semirings}
\titlecomment{{\lsuper*}This is an extended version of \cite{gustavgrabolle2021}}
\author[G.~Grabolle]{Gustav Grabolle\lmcsorcid{0000-0002-9828-7657}}	%required
\address{Institute of Computer Science, Leipzig University, 04109 Leipzig, Germany}	%required
\email{grabolle@informatik.uni-leipzig.de}  %optional
\thanks{This work was supported by Deutsche Forschungsgemeinschaft (DFG), Graduiertenkolleg 1763 (QuantLA)}	%optional

\begin{abstract}
This paper connects the classes of weighted alternating finite automata (WAFA), weighted finite tree automata (WFTA), and polynomial automata (PA).

First, we investigate the use of trees in the run semantics for weighted alternating automata and prove that the behavior of a weighted alternating automaton can be characterized as the composition of the behavior of a weighted finite tree automaton and a specific tree homomorphism if weights are taken from a commutative semiring.

Based on this, we give a Nivat-like characterization for weighted alternating automata.
Moreover, we show that the class of series recognized by weighted alternating automata is closed under inverses of homomorphisms, but not under homomorphisms.
Additionally, we give a logical characterization of weighted alternating automata, which uses weighted MSO logic for trees.

Finally, we investigate the strong connection between weighted alternating automata and polynomial automata.
We prove: A weighted language is recognized by a weighted alternating automaton iff its reversal is recognized by a polynomial automaton.
Using the corresponding result for polynomial automata, we are able to prove that the ZERONESS problem for weighted alternating automata with weights taken from the rational numbers is decidable.
\end{abstract}

\maketitle

\section{Introduction}

Non-determinism, a situation with several possible outcomes, is usually interpreted as a choice.
This view has deep historical and philosophical roots and is adapted to automata theory in the following way:
An (existential) automaton accepts if there exists at least one successful run.
However, we may as well view a situation with several possible outcomes as an obligation: A universal automaton accepts only if all possible runs are successful.
While this notion of ``universal non-determinism'' is less prominent, without further context it is as natural as the well known (existential) non-determinism.
Allowing for the simultaneous use of existential and universal non-determinism leads to the concept of alternation, such as in alternating Turing machines \cite{alternation} or alternating automata on finite \cite{BRZOZOWSKI198019}, or infinite structures \cite{MULLER1987267}. States of an alternating finite automaton (AFA) are either existential, or universal. For an existential state at least one of the outgoing runs needs to be successful, for a universal state all outgoing runs need to be successful to make the entire run successful. It is even possible to mix both modes by assigning a propositional formula over the states to each pair of state and letter. Alternating finite automata have been known for a long time. They are more succinct than finite automata and constructions like the complement, or intersection are easy for them. Due to this, they have many uses such as a stepping stone between logics and automata \cite{DeGiacomo:2013:LTL:2540128.2540252}, or in program verification \cite{Vardi1995}.

While alternating automata recognize the same class of languages as finite automata, the situation is different in the weighted setting. A weighted finite automaton (WFA) assigns a weight to each of its transitions. The weight of a run is computed by multiplying its transition weights. Finally, the automaton assigns to each input the sum over all weights of runs corresponding to this input. By this, a weighted automaton recognizes a quantitative language i.e. a mapping from the set of words into a weight structure. Depending on the weight structure used, we may view a quantitative language as a probability distribution over the words, as a cost or yield assignment, or as the likelihood or quantity of success for each input. To simultaneously allow for a multitude of interesting weight structures, weighted automata have been studied over arbitrary semirings \cite{Droste:2009:HWA:1667106}.

To adapt alternating automata into the weighted setting, it can be observed that the existence of a run in a finite automaton becomes a sum over all runs in a weighted automaton. Analogously, the demand for all runs to be successful becomes a product over all runs. More precisely, if a weighted alternating finite automaton (WAFA) is in an additive state, it will evaluate to the sum over the values of all outgoing runs. If the weighted alternating automaton is in a multiplicative state, it will evaluate to the product over the values of all outgoing runs. And again, we are able to mix both modes, this time by assigning polynomials over the states to each pair of state and letter. Weighted alternating automata over infinite words were studied in \cite{10.1007/978-3-642-03409-1_2} and in \cite{10.1007/978-3-642-24372-1_2} over finite words. While these authors focused on very specific weight structures, a more recent approach defines weighted alternating automata over arbitrary commutative semirings \cite{KOSTOLANYI20181}.

Weighted alternating automata have the same expressive power as weighted automata if and only if the semiring used is locally finite \cite{KOSTOLANYI20181}. However, for many interesting semirings such as the rational numbers, weighted alternating automata are strictly more expressive than weighted automata. While we have a fruitful framework for weighted automata, woven by results like the Nivat theorem for weighted automata \cite{droste2013weighted}, the equivalence of weighted automata and weighted rational expressions \cite{SCHUTZENBERGER1961245} and weighted restricted MSO logic \cite{DROSTE200769}, or the decidability of equality due to minimization if weights are taken from a field \cite{SCHUTZENBERGER1961245} and many more, no such results are known for weighted alternating automata. In this paper we will extend the results on weighted alternating automata by connecting them to known formalisms and thereby establishing further characterizations of quantitative languages recognized by weighted alternating automata. From there on, we will use these connections to prove interesting properties for weighted alternating automata and vice versa to translate known results for weighted alternating automata into other settings.

After a brief recollection of basic notions and notations in Section 2, Section 3 will establish several normal forms for weighted alternating automata (Lemma \ref{1Lmm}, Lemma \ref{eqLmm}) that are used as the basis of later proofs. Section 4 includes our core result (Theorem \ref{thr:WafaIffWfta}), a characterization of weighted alternating automata by the concatenation of weighted finite tree automata (WFTA) together with certain homomorphisms.
More precisely, we consider word-to-tree homomorphisms that translate words viewed as trees into trees over some arbitrary ranked alphabet.
We can show that a quantitative language is recognized by a weighted alternating automaton if, and only if there exists word-to-tree homomorphism and a weighted tree automaton such that the evaluation of the weighted alternating automata on any given word is the same as the evaluation of the weighted tree automaton on the image of the homomorphism of this word.

In Section 5 we will use this result to prove that the class of quantitative languages recognized by weighted alternating automata is closed under inverses of homomorphisms (Corollary \ref{crl:ClIHWafa}). However, we can prove the same is not true for homomorphisms in general (Lemma \ref{lmm:NClHWafa}). Since the closure under homomorphisms plays a key part in the proof of the Nivat theorem for weighted automata this prohibits a one-to-one translation of the Nivat theorem for weighted automata into the setting of weighted alternating automata. Nonetheless, we will utilize the connection between weighted alternating automata and weighted tree automata, as well as a Nivat theorem for weighted tree automata, to prove an adequate result for weighted alternating automata (Theorem \ref{thr:nivatWAFA}). This will lead us directly into a logical characterization of quantitative languages recognized by weighted alternating automata with the help of weighted restricted MSO logic for weighted tree automata (Section 6 Theorem \ref{thr:WafaIffsrMSO}). It is well known that recognizable tree languages are closed under inverses of tree homomorphisms. However, the same does not hold in the weighted setting for arbitrary commutative semirings. Section 7 gives a precise characterization of the class of semirings for which the respective class of recognizable weighted tree languages is closed under inverses of homomorphisms (Theorem \ref{thr:clOTAUIHom}). For this purpose, we will use our core theorem, and a result form \cite{KOSTOLANYI20181}.

Lastly, in Section 8, we investigate the connection between weighted alternating automata and recently introduced polynomial automata \cite{8005101} to prove the decidability of the ZERONESS and EQUALITY problems for weighted alternating automata if weights are taken from the rational numbers (Corollary \ref{crl:DecZerEqWafa}).

\section{Preliminaries}
Let $\N=\{0,1,2,\ldots\}$ denote the set of non-negative integers. For sets $M, N$ we denote the cardinality of $M$ by $|M|$, the set of subsets of $M$ by $\pot(M)$, the Cartesian product of $M$ and $N$ by $M\times N$, and the set of mappings from $M$ to $N$ by $N^M=\{f \mid f:M\rightarrow N\}$. If $M$ is finite and non-empty, it is also called an \textit{alphabet}.

For the remainder of this paper, let $\Sigma, \Gamma$ and $\Lambda$ denote alphabets. The set of all (finite) words over $\Sigma$ is denoted by $\Sigma^*$. Let $|w|$ denote the length of a word $w$ and $\Sigma^k = \{w\in \Sigma^*\mid |w|=k\}$. The unique word in $\Sigma^0$ is called the \emph{empty word} and denoted by $\varepsilon$. The concatenation of words $u,v$ is denoted by $u\cdot v$ or just $uv$. A mapping $h:\Lambda^*\rightarrow \Sigma^*$ is called a \textit{homomorphism} if $h(u\cdot v)=h(u)\cdot h(v)$ and $h$ is \textit{non-deleting} if $h(a)\neq\varepsilon$ for all $a \in \Lambda$.

A \textit{monoid} is an algebraic structure $(M, \cdot, 1)$, where $\cdot$ is a binary associative internal operation and $m \cdot 1 = m = 1\cdot m$ for all $m \in M$. A monoid is \textit{commutative} if $\cdot$ is commutative.

A \textit{semiring} is an algebraic structure $(S,+,\cdot,0,1)$, where $(S,+,0)$ is a commutative monoid, $(S,\cdot,1)$ is a monoid, $s\cdot 0 = 0 = 0\cdot s$ for all $s \in S$, and $s_3 \cdot (s_1+s_2) = s_3\cdot s_1 + s_3\cdot s_2$ and $(s_1+s_2)\cdot s_3 = s_1\cdot s_3 + s_2\cdot s_3$ for all $s_1,s_2,s_3 \in S$. A semiring is \textit{commutative} if $(S,\cdot,1)$ is commutative.

\vspace{5pt}
\centerline{\emph{For the remainder of this paper, let $S$ denote a commutative semiring.}}
\vspace{5pt}

For any set $M$, we denote $S^{M}$ by $\ser{S}{M}$. Furthermore, for $L\subseteq M$ we define the \textit{characteristic function} $\chr{L}\in \ser{S}{M}$ by $\chr{L}(w)=1$ if $w\in L$ and $\chr{L}(w)=0$ otherwise for all $w\in M$. An element $s \in \ser{S}{\Sigma^*}$ is called a \textit{$S$-weighted $\Sigma$-language}  (for short: weighted language). 

Let $X_n$ always denote a linearly ordered set with $|X_n|=n\in \N$, we refer to the $i$-th element of $X_n$ by $x_i$. Let $\pols{X_n}$ denote the \textit{semiring of polynomials} with coefficients in $S$ and commuting indeterminates $x_1,\ldots, x_n$. We say $m \in \pols{X_n}$ is a \textit{monomial} if $m=s\cdot x_1^{k_1}\cdot \ldots \cdot x_n^{k_n}$ for some $s \in S$ and $k_1,\ldots, k_n \in \N$. The \textit{degree} of $m$ is $\sum_{i=1}^{n} k_i$. A monomial $m$ is a \textit{proper monomial} if it has a non-zero degree. Each polynomial that can be written as a sum of proper monomials is called a \emph{polynomial without constants} and the set of all polynomials without constants is denoted by ${\pols{X_n}}_{\text{const}= 0}$. For $p,p_1,\ldots,p_n \in \pols{X_n}$ let $p\langle p_1,\ldots,p_n\rangle$ denote the polynomial gained from the simultaneous substitution of $x_i$ by $p_i$ in $p$ for all $1\leq i \leq n$.

Next, we give a concise collection of definitions from the topic of ranked trees that are needed in this paper.
For a more detailed introduction to trees and tree automata, we recommend \cite{tata2022}.

A \textit{ranked alphabet} is an ordered pair $(\Gamma, \rank)$, where $\rank:\Gamma\rightarrow \N$ is a mapping. Without loss of generality, we assume $X_n \cap \Gamma=\emptyset$. Moreover, let $\Gamma^{(r)}=\{\gamma\in \Gamma\mid \rank(\gamma)=r\}$ and $\rank(\Gamma)=\max\{\rank(\gamma)\mid \gamma \in \Gamma\}$. 

The set of \textit{$\Gamma$-terms over $X_n$} is the smallest set $T_{\Gamma}[X_n]$ such that $\Gamma^{(0)} \cup X_n\subseteq T_{\Gamma}[X_n]$; and $g(t_1,\ldots,t_{\rank(g)})\in T_{\Gamma}[X_n]$ for all $g \in \Gamma$ and all $t_1,\ldots,t_{\rank(g)}\in T_{\Gamma}[X_n]$. We denote $T_\Gamma (X_0)=T_\Gamma (\emptyset)$ by $T_\Gamma$. We extend $\rank$ by putting $\rank(x_i)=0$ for all $i\in\N$. If $\rank$ is clear from the context, we just write $g(t_1,\ldots, t_k)$. Moreover, we identify $g$ and $g()$ for $g\in \Gamma^{(0)} \cup X_n$. Hence, all terms $t\in T_\Gamma [X_n]$ are of the form $t=g(t_1,\ldots,t_k)$ for some $g\in \Gamma \cup X_n$ and $t_1,\ldots,t_k \in T_{\Gamma}[X_n]$.

We define $\pos:T_{\Gamma}[X_n]\rightarrow \pot(\N^*): g(t_1,\ldots,t_k)\mapsto \{\varepsilon\} \cup \bigcup_{i=1}^{k} \{i\}\cdot \pos(t_i)$. Let $t=g(t_1,\ldots,t_k)$. The mapping $\lbl{t}:\pos(t)\rightarrow\Gamma\cup X_n$ is defined by $\lbl{t}(\varepsilon)=g$; and $\lbl{t}(w)=\lbl{t_i}(v)$ if $w=iv\in \pos(t)$. We will identify $t$ and the mapping $\lbl{t}$: We write $t(w)$ to denote $\lbl{t}(w)$ and refer to terms as trees. Consequently, we have $t^{-1}(g) = \{w\in \pos(t) \mid \lbl{t}(w)=g\}$ for all $g\in T_{\Gamma}[X_n]$.
This coincides with the definition of a tree as a directed graph in the following way:
The set of vertices is $\pos(t)$, the root is $\varepsilon$, and for $u,v \in \pos(t)$ we have a $(u,v)$-edge iff $v=ui$ for some $i\in \N$.

For $t=g(t_1,\ldots,t_k),t'\in T_\Gamma [X_n]$, and $w\in \pos(t)$, the \textit{subtree of $t$ at $w$}, denoted by $t|_w$ and the \textit{substitution of $t'$ in $t$ at $w$}, denoted by $t\langle w\leftarrow t'\rangle$ are defined by $t|_\varepsilon =t$ and $t\langle \varepsilon\leftarrow t'\rangle=t'$ if $w=\varepsilon$; and $t|_w=t_i|_v$ and $t\langle w\leftarrow t'\rangle=g(t_1,\ldots, t_{i-1},t_i\langle v\leftarrow t'\rangle, t_{i+1},\ldots, t_k)$ for $w=iv\in \pos(t)$. Moreover, let $M\subseteq \pos(t)$ and $|M|=l$. We define $t\langle M\leftarrow (t'_1,\ldots,t'_l)\rangle=t\langle m_l \leftarrow t'_l\rangle \cdots \langle m_1\leftarrow t'_1\rangle$, where $m_i$ is the $i$-th element of $M$ with regard to the lexicographical order on $\N^*$. In case $t'_1=\ldots=t'_l=t'$, we abbreviate $t\langle M\leftarrow(t'_1,\ldots,t'_l)\rangle$ by $t\langle M\leftarrow t'\rangle$. If $M=t^{-1}(x_i)$, we write $t\langle x_i\leftarrow (t'_1,\ldots, t'_l)\rangle$ to denote $t\langle M\leftarrow(t'_1,\ldots, t'_l)\rangle$. Finally, let $t\langle t'_1,\ldots,t'_n \rangle=t\langle t^{-1}(x_1)\leftarrow t'_1\rangle\cdots\langle t^{-1}(x_n)\leftarrow t'_n\rangle$ denote the simultaneous substitution in trees.

We say a tree $t$ is \textit{non-deleting in $l$ variables} if it contains at least one symbol from $\Gamma$ and each of the variables $x_1,\ldots, x_l$ occurs at least once in $t$.
We say $t$ is \textit{linear in $l$ variables} if it is non-deleting in $l$ variables and each of the variables occurs at most once in $t$.
Moreover, let $\ra(t)=\sum_{i=1}^n|t^{-1}(x_i)|$ denote the number of nodes of $t$ labeled by variables. For any $r\in \N$ let $T_{\Gamma}^{(r)}(X_n)=\{t\in T_\Gamma(X_n)\mid \ra(t)=r\}$ be the set of $\Gamma$-trees with exactly $r$ positions labeled by variables.

A \textit{tree homomorphism} $h:\tlg\rightarrow T_\Lambda$ is a mapping such that for all $g \in \Gamma^{(r)}$ there exists $t_g \in T_{\Lambda}[X_r]$ with $h(g(t_1,\ldots,t_r))=t_g\langle h(t_1),\ldots, h(t_r)\rangle$ for all $t_1,\ldots,t_r \in T_\Gamma$. We will denote $t_g$ by $h(g)$, even though $t_g$ is not necessarily in $T_\Lambda$. A tree homomorphism is \textit{non-deleting} (resp. \textit{linear}) if each $h(g)$ is non-deleting (resp. linear) in $\rank(g)$ variables.
For more information on tree homomorphisms consider Paragraph 1.4 in \cite{tata2022}.

\section{Weighted alternating finite automata}

This section introduces weighted alternating finite automata (\wafa) and shows how to achieve desirable normal forms of \wafa (Lemma \ref{1Lmm}, Lemma \ref{eqLmm}). We will follow the definitions of \cite{KOSTOLANYI20181}.

A \textit{weighted alternating finite automaton} (\wafa) is a 5-tuple $\A=\wafaa$, where $Q=\{q_1,\ldots ,q_n\}$ is a finite set of states, $\Sigma$ is an alphabet, $\delta: Q \times \Sigma \rightarrow \pols{Q}$ is a transition function, $P_0 \in \pols{Q}$ an initial polynomial, and $\tau: Q \rightarrow S$ a final weight function.

Let $\A=\wafaa$ be a \wafa. Its \textit{state behavior} $\bhvs{\A}: Q\times\Sigma^*\rightarrow S$ is the mapping defined by \[\bhvs{A}(q,w) = \begin{cases} \tau(q) &\text{if } w=\varepsilon , \\ \delta(q,a)\big\langle\bhvs{\A}(q_1,v),\ldots, \bhvs{\A}(q_n,v)\big\rangle &\text{if } w=av \text{ for } a\in \Sigma\enspace. \end{cases}\]
Usually, we will write $\bhvs{\A}_q(w)$ instead of $\bhvs{\A}(q,w)$. 
Now, the \textit{behavior} of $\A$ is the weighted language $\bhv{\A}:\Sigma^*\rightarrow S$ defined by \[\bhv{A}(w)= P_0\big\langle\bhvs{\A}_{q_1}(w),\ldots,\bhvs{\A}_{q_n}(w)\big\rangle\enspace.\]
A weighted language $s$ is \textit{recognized} by $\A$ if and only if $\bhv{\A}=s$. Two \wafa are said to be \textit{equivalent} if they recognize the same weighted language.

Within the scope of this paper, we define a weighted finite automaton as follows: 
A \emph{weighted finite automaton} (\wfa) is a \wafa $\A = (\{q_1, \ldots, q_n\}, \Sigma, \delta, P_0, \tau)$ where $P_0 = \sum_{j=1}^{n} s_j\cdot q_j$ and $\delta(q_i,a)= \sum_{j=1}^{n} s^{a}_{ij}\cdot q_j$ for all $1\leq i \leq n$, $a \in \Sigma$.
This definition coincides with the well known definition cf. \cite{HWAc5},\cite{DROSTE200769}. For example $P_0 = q_1 + 3 \cdot q_4$ corresponds to initial weight $1$ in $q_1$, initial weight $3$ in $q_4$, and initial weight $0$ in all other states;
or $\delta(q_2, a) = 3 \cdot q_2 + 4 \cdot q_3$ corresponds to a situation where state $q_2$ has an $a$-loop with weight $3$ and an $a$-transition to state $q_3$ with weight $4$ (and vice versa). Lastly, $\tau(q_2) = 2$ means state $q_2$ has final weight $2$.

Let $\A$ be a \wafa with weights taken from $S$.
We define $M_{(q,a)}$ as the \emph{set of monomials that appear in $\delta(q,a)$} and $M_0$ as the \emph{set of monomials that appear in $P_0$}.
An element $s \in S$ is called a \emph{coefficient in $\A$} if $s$ is the coefficient of a monomial in $M_0$, or the coefficient of a monomial in $M_{(q,a)}$ for some $q \in Q, a \in \Sigma$.
Similarly, $s$ is called a \emph{constant in $\A$} if $P_0(0, \ldots, 0) = s$, or $\delta(q,a)(0, \ldots, 0) = s$ for some $q\in Q, a \in \Sigma$. 

We say $\A$ is \textit{nice} if it has the following properties:
\begin{enumerate}[(i)]
\item $\delta(q,a)$ is a finite sum of pairwise distinct monomials of the form $s\cdot q_1^{k_1}\cdot\ldots \cdot q_n^{k_n}$ for all $q\in Q, a\in \Sigma$,
\item all monomials in $P_0$ and $\delta$ are proper,
\item $P_0 =q_1$.
\end{enumerate}
Moreover, we say that $\A$ is \textit{purely polynomial} if: 
\begin{enumerate}
  \item[(iv)] all monomials (in $P_0$ and $\delta$) have coefficient $1$.
\end{enumerate}

We want to show that we can always assume a \wafa to be nice and purely polynomial.

\begin{lem}\label{1Lmm} For each \wafa $\A$ there exists an equivalent \wafa $\A'$ such that (i)-(iv) hold for $\A'$.
The construction of $\A'$ is effective.
\end{lem}

\begin{proof}
Let $\A=\wafaa$ be a \wafa.
\begin{enumerate}[(i)]
\item Since $\cdot$ is distributive and commutative in $\pols{Q}$ there exists an equivalent \wafa $\A'$ such that (i) holds. 

\item Assume (i) holds for $\A=\wafaa$. We define a \wafa $\A'=(Q',\Sigma, \delta', P_0', \tau')$ which includes a new state $q_c$ for each constant $c$ in $\A$. Furthermore, $\delta'$ and $P_0'$ are as $\delta$ and $P_0$, respectively, but each occurrence of each constant $c$ is replaced by $q_c$. Moreover, $\delta'(q_c,a)=q_c$ for all $a\in \Sigma$ and $\tau'(q_c)=c$. There is a finite number of constants in $\A$. Thus, $\A'$ is a \wafa.  It is easy to see that $\A$ and $\A'$ are equivalent and that (i)-(ii) hold for $\A'$.

\item Assume (i)-(ii) hold for $\A$. Due to Lemma 6.3 of \cite{KOSTOLANYI20181}, there exists an equivalent \wafa $\A'$ such that (i)-(iii) hold for $\A'$.
The construction of $\A'$ is straightforward: we add a new state $q$ to $\A$ with $\delta(q, a) = P_0\langle \delta(q_1,a), \ldots, \delta(q_n, a)\rangle$ and $\tau(q) = P_0\langle t(q_1), \ldots, \tau(q_n)\rangle$.
We rename the states of $\A$ such that $q$ becomes the new $q_1$.
Lastly, we set $P_0 = q_1$.

\item Assume that (i)-(iii) hold for $\A$. We define \[Q'=Q \cup \{q_s \mid s \text{ is a coefficient in }\A\}.\] We may assume that the two sets forming this union are disjoint. Furthermore, let $\delta'$  be defined by
\[
\delta'(q,a) = \begin{cases}\sum\limits_{s\cdot{q_1}^{k_1} \ldots {q_n}^{k_n} \in M_{(q,a)}} {q_1}^{k_1} \ldots {q_n}^{k_n} \cdot q_s &\text{if } q \in Q \text{ and}\\ q &\text{otherwise} \end{cases}\] for all $q \in Q'$, $a \in \Sigma$. Moreover, let $\tau'$  be defined by
\[
\tau'(q) = \begin{cases}\tau(q) &\text{if } q \in Q \text{ and}\\ s &\text{if } q=q_s\end{cases}
\] for all $q \in Q'$. Consider the \wafa $\A'=(Q',\Sigma,P_0,\delta',\tau')$. It is easy to see that $\bhv{\A}=\bhv{\A'}$ and that every monomial in $\A'$ has $1$ as coefficient. 

If the appropriate order on $Q'$ is chosen, properties (i)-(iii) hold for $\A'$, too.\qedhere
\end{enumerate} 
\end{proof}

We would like to point out: while the construction described in the proof of Property (iv) also works for \wfa the resulting automaton does not have to be a \wfa.

In \cite{KOSTOLANYI20181} the transition function and the initial polynomial are not allowed to contain constants. This corresponds to the property that runs are not allowed to terminate before the entire word is read. Since it will help with several constructions, we allowed constants in our definition. Nevertheless, as Lemma \ref{1Lmm} (ii) shows, the introduction of constants does not increase expressiveness since it is possible to simulate terminating transitions by ``deadlock''-states.

We say a \wafa $\wafaa$ is \textit{equalized} if all monomials occurring in $\delta$ have the same degree.

\begin{lem}\label{eqLmm} For each \wafa $\A$ there exists an equivalent, equalized and nice\wafa $\A'$.
  The construction of $\A'$ is effective.
\end{lem}

\begin{proof} Let $\A=(Q,\Sigma,\delta,q_1,\tau)$ be a nice \wafa and $d$ the maximum degree of monomials occurring in $\delta$. Let $q_{n+1}$ be a new state and 
\[
\delta'(q,a) = \begin{cases}\sum\limits_{s{q_1}^{k_1} \ldots {q_n}^{k_n} \in M_{(q,a)}} s\cdot {q_1}^{k_1} \ldots {q_n}^{k_n} \cdot q_{n+1}^{d-\sum_{u=1}^n k_u} &\text{if } q \in Q \text{ and}\\ q_{n+1}^d &\text{otherwise} \end{cases}\] for all $q \in Q'$, $a \in \Sigma$. As well as, \[\tau'(q) = \begin{cases}\tau(q) &\text{if } q \in Q \text{ and}\\ 1 &\text{if } q=q_{n+1}\end{cases}
\] for all $q \in Q'$. Clearly, $\A'=(Q\cup\{q_{n+1}\},\Sigma,\delta',q_1,\tau')$ is equalized. Also, $\A$ and $\A'$ are equivalent and $\A'$ is nice.

Due to the form of the constructed polynomials and since the initial polynomial remains unchanged, we may assume that $\A'$ is nice.
Moreover, we didn't change the coefficients, thus $\A'$ is purely polynomial, if $\A$ was purely polynomial.
\end{proof}

Nice \wafa can be represented in the following way: As usual we depict each state by a circle. Then, each monomial $s\cdot q_1^{k_1}\ldots q_n^{k_n}$ in $\delta(q_i,a)$ is represented by a multi-arrow which is labeled by $a:s$, begins in $q_i$, and has $k_j$ heads in $q_j$ for all $1\leq j\leq n$, respectively. In case a multi-arrow has more than one head, we join these heads by a \tikz{\draw[draw= grau2, fill=grau2] (0,0) circle (1.5pt);}. If $s=1$, we omit the $s$-label. If $s=0$, we omit the complete multi-arrow. The initial polynomial is represented analogously. The final weights are represented as usual. Note that the multi-arrows can be viewed as a parallel or simultaneous transitions and that this representation coincides with the usual representation if the automaton is a $\wfa$. Consider the following example:

\begin{exa}\label{xmp:serp} Let $S=(\N,+,\cdot,0,1)$, $\Sigma=\{a,b\}$, and $s$ the weighted language \[\begin{array}{llcl} s: &\Sigma^*&\rightarrow &S:\\
&w &\mapsto &\begin{cases}{(2^{j})}^{2^i} &\text{if } w=a^ib^j \enspace, \\ 0 &\text{otherwise.}\end{cases}\end{array} \]
We consider the \wafa $\A=(\{q,p\},\Sigma,P_0,\delta,\tau)$, defined by:
\[\begin{array}{lcllclclcllcl}
P_0 &=&q & & & &\hspace{2cm} &\delta(q,a)&=&{q}^2 &\delta(q,b)&= &p\\
\tau(q)&=&1 &\tau(p)&=&2 & &\delta(p,a)&=&0 &\delta(p,b)&=&2\cdot p
\end{array}\]
A depiction of this automaton can be seen in Figure \ref{rep1}. One can check that $\bhv{\A}=s$, for example:
\[\begin{array}{clclcl}
& \bhv{\A}(aabb) & = & q\big\langle\bhvs{\A}_q(aabb),\bhvs{\A}_p(aabb)\big\rangle & & \\
= & \bhvs{\A}_q(aabb) & = & q^2\big\langle\bhvs{\A}_q(aabb),\bhvs{\A}_p(aabb)\big\rangle & & \\
= & \big(\bhvs{\A}_q(abb)\big)^2 & = & \big(\bhvs{\A}_q(bb)\big)^{2\cdot 2} & = & \big(\bhvs{\A}_p(b)\big)^{2\cdot 2}\\
= & \big(2\cdot\bhvs{\A}_p(\varepsilon)\big)^{2\cdot 2} & = & \big(2\cdot\tau(p)\big)^{2\cdot 2} & = & {(2^2)}^{2^2}
\end{array}\]
\begin{figure}
\centering
\begin{minipage}{.5\textwidth}
\centering
\begin{tikzpicture}[every edge/.style={main edge},
          every loop/.style={main edge}, 
          every initial by arrow/.style={main edge, initial distance= 10pt},
          every accepting by arrow/.style={main edge, accepting distance= 10pt}]
  \def\x{1.5}
  \def\y{1.5}

  \node[main node,initial by arrow, initial where=above, initial text=,accepting by arrow, accepting where= below, accepting text=$1$ ] (q1) at (-1*\x,0*\y) {$q$};
  \node[main node,accepting by arrow, accepting where= right, accepting text=$2$ ] (q2) at (1*\x,0*\y) {$p$};
  \phantom{\node[main node,accepting by arrow, accepting where= below, accepting text=$1$ ] (q3) at (0*\x,-.8*\y) {$h_1$};}

  \path (q1) edge node[edge label, above] {$b$} (q2)
             edge[out=180, in=150, looseness=9] node[edge label, xshift=-5pt, yshift=-5pt] {$a$} (q1)
             edge[out=180, in=210, looseness=9] (q1)

        (q2) edge[out=-90, in=-60, looseness=9] node[edge label, yshift=-5pt] {$b$:$2$} (q2);
  \draw[draw= grau2, fill=grau2] ([rotate around={0:(q1)}]q1.west) circle (1.5pt);

\end{tikzpicture}
\caption{Representation of $\A$}\label{rep1}
\end{minipage}%
\begin{minipage}{.5\textwidth}
\centering
\begin{tikzpicture}[every edge/.style={main edge},
          every loop/.style={main edge}, 
          every initial by arrow/.style={main edge, initial distance= 10pt},
          every accepting by arrow/.style={main edge, accepting distance= 10pt}]
  \def\x{1.5}
  \def\y{1.5}

  \node[main node,initial by arrow, initial where=above, initial text=,accepting by arrow, accepting where= below, accepting text=$1$ ] (q1) at (-1*\x,0*\y) {$q$};
  \node[main node,accepting by arrow, accepting where= right, accepting text=$2$ ] (q2) at (1*\x,0*\y) {$p$};
  \node[main node,accepting by arrow, accepting where= below, accepting text=$1$ ] (q3) at (0*\x,-.8*\y) {$h_1$};

  \path (q1) edge node[edge label, above] {$b$} (q2)
             edge[out=180, in=150, looseness=9] node[edge label, xshift=-5pt, yshift=-5pt] {$a$} (q1)
             edge[out=180, in=210, looseness=9] (q1)
             edge[out=0, in = 120] (q3)

         (q3) edge[out=180, in=150, looseness=9] node[edge label, xshift=-5pt, yshift=-5pt] {$\Sigma$} (q3)
         edge[out=180, in=210, looseness=9] (q3)

        (q2) edge[out=-90, in=-60, looseness=9] node[edge label, yshift=-5pt] {$b$:$2$} (q2)
             edge[out=-90, in=60] (q3);
  \draw[draw= grau2, fill=grau2] ([rotate around={0:(q1)}]q1.west) circle (1.5pt);
  \draw[draw= grau2, fill=grau2] ([rotate around={0:(q1)}]q1.east) circle (1.5pt);
  \draw[draw= grau2, fill=grau2] ([rotate around={0:(q3)}]q3.west) circle (1.5pt);
  \draw[draw= grau2, fill=grau2] ([rotate around={0:(q2)}]q2.south) circle (1.5pt);

\end{tikzpicture}
\caption{Representation of equalized, nice $\A$}\label{rep2}
\end{minipage}
\end{figure}
\end{exa}

Two-mode alternating automata are a subclass of alternating automata.
In the weighted setting two-mode alternating automata are defined as follows (cf. \cite{KOSTOLANYI20181}):
each state $q$ is either existential ($\delta(q,a)=\sum_{i=1}^n s^{a}_i\cdot q_i$ for all $a\in \Sigma$) or universal ($\delta(q,a)=s^{a}\cdot q_1^{k_1}\cdot\ldots\cdot q_n^{k_n}$ for all $a\in \Sigma$). In \cite{KOSTOLANYI20181} run semantics were defined for two-mode alternating automata. However, they can be defined for alternating automata with mixed states, in an analogue fashion. Here, we give the definition of runs for nice \wafa. For one, as seen above, we can transform every \wafa into a nice one. For another, it is only a technicality to define runs for arbitrary \wafa, based on the definition below.

The idea of reading monomials as multi-transitions was already introduced above. Having this in mind, a run tree is a tree labeled by states such that: $(i)$ Our run begins in the initial state, $(ii)$ if a vertex at depth $k$ is $q$ labeled, then the labels of its children fit the states of an $a_k$ labeled multi-transition beginning in state $q$, and $(iii)$ our run halts after $n$ steps.

More formally, if $\A=(Q,\Sigma,\delta,q_1,\tau)$ is a nice \wafa, we define the ranked alphabet $\Gamma_\A=\{(q,x)\in Q \times \pols{Q}\mid x=\tau(q) \text{ or } x \in M_{(q,a)} \text{ for some } a \in \Sigma\}$ where $\rank(q,x)$ is the rank of $x$ as a polynomial. A run over $w=a_1\ldots a_n$ in $\A$ is a $\Gamma_\A$-tree $\tr$ with the following properties:
\begin{enumerate}[(i)]
  \item If $|p|=0$, then $\tr(p)=(q_1,x)$, for some $x$.
  \item If $|p|<n$ and $\tr(p)=(q,s\cdot q_1^{k_1}\cdot\ldots\cdot q_n^{k_n})$, then $s\cdot q_1^{k_1}\cdot\ldots\cdot q_n^{k_n}\in M_{(q,a_{|p|+1})}$ and for all $\sum_{u=1}^{l-1}k_u \leq i < \sum_{u=1}^{l}k_u$ we have $\tr(pi)=(q_l,x)$ for some $x$.
  \item If $|p|= n$, then $\rank(\tr(p))=0$.
\end{enumerate}

The weight of a run tree $\tr$ is defined by \[\weight(\tr)=\hspace{40pt}\prod_{\mathclap{(q,s\cdot q_1^{k_1}\cdot\ldots\cdot q_n^{k_n})\in \tr(\pos(\tr))}} s \enspace.\] Note that the final weights are accounted for, since they are the labels of the leaves of our run tree.

And as usual we have (Theorem 5.15 \& Theorem 5.17 in \cite{KOSTOLANYI20181}): \[\bhv{\A}(w)=\sum_{\tr \text{ run over } w}\weight(\tr)\enspace.\]

\begin{exa} Later we will use the connection between \wafa and trees heavily. For a better understanding of this connection we want to point out the following:

First off, let us observe that $|M_{(q,a)}|=1$ if $\delta(q,a)$ is a monomial. Consequently, if $\delta(q,a)$ is a monomial for all $q\in Q, a\in \Sigma$, then every word has a unique run. We call a \wafa with this property \textit{universal}.
A two-mode \wafa where every state is universal is universal and vice versa.
Universality for \wafa is as determinism for finite automata. However, in general universal \wafa and \wfa are incomparable it terms of expressive power (cf. Corollary 3 in \cite{10.1007/978-3-642-24372-1_2}).

Second, we consider the non-universal automaton from Figure \ref{fig:nonDetWAFA}. This automaton is not universal since $\delta(q,a)= q + p$. As a consequence, we may have several runs, whenever the letter $a$ is to be processed in state $q$ we can choose one of two possible children in our tree. Figure \ref{fig:runsOfAWAFA} shows these run-trees for the input $aba$. This connection becomes more clear, if we observe the behavior of the automaton on $aba$:\[q \overset{a}{\rightarrow} q+p \overset{b}{\rightarrow} (q\cdot p) + (q) \overset{a}{\rightarrow} ((q+p) \cdot p)+ (q+p) = q\cdot p + p\cdot p + q + p\] Clearly, every monomial in the final polynomial corresponds to one of the run-trees.
Moreover, the weight of a run-tree is the product of the coefficients in its nodes and its leaves correspond to the final weights of their state.
Thus, we get that the sum of all run-weights is the behavior of the automaton.
\end{exa}

\begin{figure}
\begin{minipage}[b]{.4\textwidth}%
\centering
\begin{tikzpicture}[every edge/.style={main edge},
          every loop/.style={main edge}, 
          every initial by arrow/.style={main edge, initial distance= 10pt},
          every accepting by arrow/.style={main edge, accepting distance= 10pt}]
  \def\x{1.5}
  \def\y{1.5}

  \node[main node] (q1) at (-1*\x,0*\y) {$q$};
  \node[main node] (q2) at (1*\x,0*\y) {$p$};
  \node[inner sep=2pt] (h1) at (-1*\x-.7,0*\y+.7) {\phantom{1}};
  \node[inner sep=2pt] (h2) at (-1*\x-.7,0*\y-.7) {1};
  \node[inner sep=2pt] (h3) at (1*\x+.7,0*\y-.7) {2};

  \path (q1) edge[bend right] node[edge label, below] {$a$} (q2)
             edge[out=45, in=90, looseness=9] node[edge label, xshift=12pt, yshift=0pt] {$b$} (q1)
             edge[out=45, in=135] (q2)
             edge[out=247.5, in=292.5, looseness=8] node[edge label, below] {$a$} (q1)

        (q2) edge node[edge label, above] {$b$} (q1)
            edge[out=247.5, in=292.5, looseness=8] node[edge label, below] {$a$} (q2)

        (h1) edge (q1)
        (q1) edge (h2)
        (q2) edge (h3);
  \draw[draw= grau2, fill=grau2] ([rotate around={45:(q1)}]q1.east) circle (1.5pt);

\end{tikzpicture}
\caption{A non-universal WAFA}\label{fig:nonDetWAFA}
\end{minipage}%
\begin{minipage}[b]{.6\textwidth}%
  \centering
  \begin{tikzpicture}[every edge/.style={main edge}]
    \def\x{2.3}

    \node[inner sep=0pt] (111) at (0,0) {$(q,q)$};
    \node[inner sep=0pt] (121) at (0,-.7) {$(q,q\cdot p)$};
    \node[inner sep=0pt] (131) at (-.5,-1.4) {$(q, q)$};
    \node[inner sep=0pt] (132) at (.5,-1.4) {$(p, p)$};
    \node[inner sep=0pt] (141) at (-.5,-2.1) {$(q, 1)$};
    \node[inner sep=0pt] (142) at (.5,-2.1) {$(p, 2)$};

    \draw (111) edge[-] (121)
          (121) edge[-] (131)
          (121) edge[-] (132)          
          (131) edge[-] (141)
          (132) edge[-] (142);

    \node[inner sep=0pt] (211) at (0+\x,0) {$(q,q)$};
    \node[inner sep=0pt] (221) at (0+\x,-.7) {$(q,q\cdot p)$};
    \node[inner sep=0pt] (231) at (-.5+\x,-1.4) {$(q, p)$};
    \node[inner sep=0pt] (232) at (.5+\x,-1.4) {$(p, p)$};
    \node[inner sep=0pt] (241) at (-.5+\x,-2.1) {$(p, 2)$};
    \node[inner sep=0pt] (242) at (.5+\x,-2.1) {$(p, 2)$};

    \draw (211) edge[-] (221)
          (221) edge[-] (231)
          (221) edge[-] (232)          
          (231) edge[-] (241)
          (232) edge[-] (242);

    \node[inner sep=0pt] (311) at (0+2*\x,0) {$(q,p)$};
    \node[inner sep=0pt] (321) at (0+2*\x,-.7) {$(p,q)$};
    \node[inner sep=0pt] (331) at (0+2*\x,-1.4) {$(q, q)$};
    \node[inner sep=0pt] (341) at (0+2*\x,-2.1) {$(q, 1)$};

    \draw (311) edge[-] (321)
          (321) edge[-] (331)         
          (331) edge[-] (341);

    \node[inner sep=0pt] (411) at (0+3*\x,0) {$(q,p)$};
    \node[inner sep=0pt] (421) at (0+3*\x,-.7) {$(p,q)$};
    \node[inner sep=0pt] (431) at (0+3*\x,-1.4) {$(q, p)$};
    \node[inner sep=0pt] (441) at (0+3*\x,-2.1) {$(q, 2)$};

    \draw (411) edge[-] (421)
          (421) edge[-] (431)         
          (431) edge[-] (441);

  \end{tikzpicture}
  \caption{Runs of \wafa from Figure \ref{fig:nonDetWAFA} on $aba$}\label{fig:runsOfAWAFA}
\end{minipage}
\end{figure}

Due to the bound on the growth of series recognized by \wfa, we can see that $s$ from Example \ref{xmp:serp} is not recognizable by a \wfa. Thus, \wafa are more expressive than \wfa when weights are taken from the non-negative integers. However, this is not the case for every semiring. A semiring $S$ is \textit{locally finite} if for every finite $X\subseteq S$ the generated subsemiring $\langle X \rangle$ is finite.
The following result characterizes semirings on which \wafa and \wfa are equally expressive:

\begin{thmC}[{\cite[Theorem 7.1]{KOSTOLANYI20181}}]\label{thr:WAFAiffWFAiffLocFin} The class of $S$-weighted $\Sigma$-languages recognizable by \wafa and the class of $S$-weighted $\Sigma$-languages recognizable by \wfa are equal if and only if $S$ is locally finite.
\end{thmC}

\section{A characterization of \wafa via weighted finite tree automata}
Our central result Theorem \ref{thr:WafaIffWfta} is included in this section, as well as the definition of weighted finite tree automata.

The connection between alternating automata and trees was utilized before: Already in the non-weighted settings trees were used to define runs of alternating automata.
As seen above, this is possible for \wafa too.
We want to strengthen this connection by the use of tree automata and tree homomorphisms.
In order to do so, we need some additional definitions.

An element $r\in \ser{S}{\tlg}$ is called a \textit{($S$-weighted) tree language}. A \textit{weighted finite tree automaton} (\wfta) is a 4-tuple $\A=\wftaa$, where $Q=\{q_1,\ldots,q_n\}$ is a finite set of states, $\Gamma$ is a ranked alphabet, $\delta=(\delta_k\mid 1\leq k\leq\rank(\Gamma))$ is a family of transition functions $\delta_k:\Gamma^{(k)}\rightarrow S^{Q^k\times Q}$, and $\lambda: Q \rightarrow S$ a root weight function. 

If $k$ is clear from the context, we will denote tuples $(p_1,\ldots,p_k)$ by $\overbar{p}$. Moreover, since $k$ in $\delta_k(g)$ is clear from $g$, we will denote $\delta_k(g)$ by $\delta_g$.

Let $\A=\wftaa$ be a \wfta. Its \textit{state behavior} $\bhvs{\A}: Q\times \tlg\rightarrow S$ is the mapping defined by \[\bhvs{A}\big(q,g(t_1,\ldots,t_k)\big) =\sum\limits_{\overbar{p}\in Q^k}\delta_g(\overbar{p},q) \cdot \prod_{i=1}^{k} \bhvs{\A}(p_i,t_i)\enspace.\]
Usually, we will write $\bhvs{\A}_q(t)$ instead of $\bhvs{\A}(q,t)$. 
Now, the \textit{behavior} of $\A$ is the weighted tree language $\bhv{\A}:\tlg\rightarrow S$ defined by \[\bhv{A}(t)= \sum_{i=1}^{n}\lambda(q_i)\cdot \bhvs{\A}_{q_i}(t)\enspace.\]
A weighted tree language $s$ is \textit{recognized} by $\A$ if and only if $\bhv{\A}=s$.

It is well known that a word over $\Sigma$ can be represented as a $1$-ary tree: Each letter of $\Sigma$ is given rank one and a new end-symbol $\#$ of rank zero is added. Then $w_0w_1\ldots w_n$ translates to the tree $w_0(w_1(\ldots w_n(\#)\ldots))$. Here, we want to represent words as full $r$-ary trees for any arbitary $r \in \N$. Given an alphabet $\Sigma$ and $r\geq 1$, we define the ranked alphabet $\Sigma_\#^{r}=\Sigma \cup \{\#\}$ with $\rank(\#)=0$ and $\rank(a)=r$ for all $a \in \Sigma$. For all $w\in \Sigma^*$ the tree $\tw\in T_{\Sigma_\#^r}$ is defined by $t^r_\varepsilon=\#$; and $t^r_w=a(t^r_v,\ldots,t^r_v)$ if $w=av$ with $a\in \Sigma$. We call $h^r:\Sigma^* \rightarrow T_{\Sigma_\#^r}: w\mapsto t_w^r$ the \textit{generic tree homomorphism (of rank $r$)}. The case $r=1$ is special since for all $t\in T_{\Sigma_\#^1}$ there exists $w\in \Sigma^*$ such that $t=t_w^1$. Therefore, if clear from the context, we will identify $\Sigma$ and $\Sigma_\#^1$, $\Sigma^*$ and $T_{\Sigma_\#^1}$, as well as $w$ and $t_w^1$. It is well known that a weighted $\Sigma$ language is recognizable by a \wfa over $\Sigma$ if and only if it is recognizable by a \wfta over $\Sigma_\#^1$.

The key observation is that the behavior of a \wafa $\A$ on $w$ can be characterized by the behavior of a \wfta on $\tw$ where $r$ is the degree of polynomials in an equalized version of $\A$. Even more, the behavior of a \wfta on $h(w)$ (where $h$ is a tree homomorphism) can be characterized by the behavior of a \wafa on $w$.

\begin{lem}\label{lmm:d1WafaIffWfta} If $s\in \ser{S}{\Sigma^*}$ is recognized by a \wafa, then $s=\bhv{\B}\circ h^r$ for some \wfta $\B$ and the generic homomorphism $h^r$ for some $r\in \N$.
\end{lem}

\begin{proof} Assume $s$ is recognized by a \wafa. Due to Lemma \ref{1Lmm} and Lemma \ref{eqLmm}, we may assume that $s$ is recognized by a nice and equalized \wafa $\A=(Q,\Sigma,\alpha,P_0,\tau)$. Let $r$ be the unique degree of monomials in $\A$. Our goal is to define a \wfta $\B$ with $\bhv{\A}=\bhv{\B}\circ h^r$. To this purpose, we observe that the equalized $\A$ can be viewed as a \wfta over a ranked alphabet where each letter has rank $r$ and where each multi-arrow in the representation of $\A$ corresponds to one transition in $\B$. Formally, we define the \wfta $\B=(Q,\Sigma^r_{\#},\beta,\lambda)$ with:
\[
\begin{array}{lcl}
  \lambda &= &\chr{\{q_1\}}\\
  \beta_\#(\varepsilon,q)&=&\tau(q)\enspace ,\\
  \beta_a(\overbar{p},q)&=&\begin{cases}s &\text{if } p_1 \leq \ldots \leq p_r \text{ according to the order on } Q \text{ and }\\ &\text{\phantom{if }}s\cdot p_1\cdot\ldots\cdot p_{r}\in M_{(q,a)}\enspace, \text{ and}\\
                        0 &\text{otherwise}\enspace ,\end{cases}
\end{array}
\]
for all $q \in Q$ and $\overbar{p} \in Q^r$.

Note that the order of the $p_1, \ldots, p_r$ needs to be fixed, otherwise many runs of the \wfta may correspond to one run of the \wafa.

Now, we will show that $\bhvs{\A}_q(w)=\bhvs{\B}_q(\tw)$ for all $q \in Q$ and $w \in \Sigma^*$ by induction on $|w| \in \N$. If $|w|=0$, then $w=\varepsilon$. Thus, $\bhvs{\A}_q(w)=\tau(q)=\beta_\#(\varepsilon,q)=\bhvs{\B}_q(t^r_w)$ for all $q\in Q$. Assume there exists $l\in \N$ such that the claim holds for all $w' \in \Sigma^{l}$ and all $q\in Q$. For the induction step let $w=aw'\in \Sigma\cdot \Sigma^{l}$. For $\overbar{p}\in Q^k$ let $k_i(\overbar{p})$ denote the number of $q_i$'s in $\overbar{p}$. Then, we get
  \[\begin{array}{cl}
   &\bhvs{\A}_q(w)\\
  = &\alpha(q,a)\langle\bhvs{\A}_{q_1}(w'),\ldots, \bhvs{\A}_{q_n}(w')\rangle\\
  \overset{\mathclap{\textup{(IH)}}}{=} &\alpha(q,a)\langle\bhvs{\B}_{q_1}(t^r_{w'}),\ldots, \bhvs{\B}_{q_n}(t^r_{w'})\rangle\\
  = &\sum\limits_{s\cdot q_1^{k_1}\ldots q_n^{k_n}\in M_{(q,a)}} s\cdot \prod\limits_{i=1}^n {\big(\bhvs{\B}_{q_i}(t^r_{w'})\big)}^{k_i}\\
  = &\sum\limits_{s\cdot q_1^{k_1}\ldots q_n^{k_n}\in M_{(q,a)}} \beta_a(\overbrace{\underbrace{q_1,\ldots, q_1}_{k_1\text{ times}}, \ldots , \underbrace{q_{n},\ldots, q_{n}}_{k_n\text{ times}}}^{\sum k_i=r\text{ times}}, q) \cdot \prod\limits_{i=1}^n {\big(\bhvs{\B}_{q_i}(t^r_{w'})\big)}^{k_i}\\
  =&\sum\limits_{\overbar{p}\in Q^r}\beta_a(\overbar{p},q) \cdot \prod\limits_{i=1}^n {\big(\bhvs{\B}_{q_i}(t^r_{w'})\big)}^{k_i(\overbar{p})}\\
  =&\bhvs{\B}_{q}\big(a(\underbrace{t^r_{w'},\ldots,t^r_{w'}}_{r\text{ times}})\big)\\
  =&\bhvs{\B}_{q}(\tw)
  \end{array}\] for all $q\in Q$. Which completes our induction.

Since $\A$ is nice and thus $P_0=q_1$, we consequently have \[\bhv{\A}(w) = \bhvs{\A}_{q_1}(w) = \bhvs{\B}_{q_1}(\tw) = \sum_{i=1}^{n}\lambda(q_i)\cdot \bhvs{\B}_{q_i}(\tw) = \bhv{\B}(\tw)\] for all $w\in \Sigma^*$. Since $h^r(w)=t_w^r$, this finishes our proof.
\end{proof}

The following example illustrates this connection between \wafa and \wfta.

\begin{exa} We consider the automaton $\A$ from Example \ref{xmp:serp}. It is easy to construct the corresponding \wfta $\B=(Q,\Gamma,\beta,\lambda)$ from the equalized version $\A'$ (Figure \ref{rep2}). First, we copy the set of states (in order) $Q=\{q,p,h_1\}$. Since the maximum degree of polynomials in $\A$ was $2$ we get $\Gamma=\{a^{(2)}, b^{(2)},{\#}^{(0)}\}$. The root weight function corresponds to the initial weights. However, $\A'$ is nice and thus $\lambda = \chr{\{q\}}$. The transition weight functions $\beta_a$ and $\beta_b$ can be defined using the multi arrows in Figure \ref{rep2}. For example, the $b$-labeled multi arrow in the middle corresponds to $\beta_b(ph_1,q)=1$. Finally, the final weights in $\A'$ are captured by $\beta_\# (\varepsilon,q)=1, \beta_\# (\varepsilon,h_1)=1$, and $\beta_\# (\varepsilon,p)=2$. The only non-zero run on $t^2_{ab}$ can be seen in Figure \ref{rep3}.

\begin{figure}
\centering
\begin{tikzpicture}
    \node (hash) at (-5,-1.8) {\textbf{\#}};
    \node (b) at (-3,-1.8) {\textbf{b}};
    \node (a) at (-1,-1.8) {\textbf{a}};
    
    \node (root) at (2,0) {$1$};
    \node (1) at (0,0) {$q$};
    \node (2) at (-2,0.7) {$q$};
    \node (3) at (-2,-0.7) {$q$};
    \node (4) at (-4,1.05) {$p$};
    \node (5) at (-4,0.35) {$h_1$};
    \node (6) at (-4,-0.35) {$p$};
    \node (7) at (-4,-1.05) {$h_1$};
    \node (8) at (-6,1.05) {$2$};
    \node (9) at (-6,0.35) {$1$};
    \node (10) at (-6,-0.35) {$2$};
    \node (11) at (-6,-1.05) {$1$};

    \path (11) edge[->] (7)
          (10) edge[->] (6)
          (9) edge[->] (5)
          (8) edge[->] (4)
          (7) edge[->] node[midway, below] {$2$} (3)
          (6) edge[->] (3)
          (5) edge[->] node[midway, below] {$2$} (2)
          (4) edge[->] (2)
          (3) edge[->] node[midway, below] {$1$}(1)
          (2) edge[->] (1)
          (1) edge[->] (root);
\end{tikzpicture}
\caption{Run of translated \wfta on $t^2_{ab}$}\label{rep3}
\end{figure}
\end{exa}

\begin{lem}\label{lmm:d2WafaIffWfta} Let $\B=\wftaa$ a \wfta and $h:\Sigma^*\rightarrow T_\Gamma$ a tree homomorphism. Then $\bhv{\B}\circ h$ is recognized by a \wafa.
\end{lem}

\begin{proof}
Assume $s=\bhv{\B}\circ h \in \ser{S}{\Sigma^*}$, where $\B=\wftaa$ is a \wfta with $|Q|=n$ and $h: \Sigma^* \rightarrow \tlg$ a tree homomorphism. We want to construct a \wafa $\A$ such that $\bhv{\A}=s$.

If $h$ is the generic homomorphism, we can define $\delta'(q,a)=\sum_{\overbar{p}\in Q^{r}} \delta_a(\overbar{p},q)$ and use the same proof as in the first direction. However, we want to prove this for arbitrary homomorphisms. To achieve this, we give some additional definitions.

Under $h$, each letter becomes a tree. Nonetheless, we are not interested in the structure of $h(a)$, but want to handle it as if it is a ranked letter. Therefore, we use $h(a)\langle x_1\leftarrow(p_1,\ldots,p_{r})\rangle$ to disambiguate its $r=\ra(h(a))$ variables. Furthermore, we extend the family of transition functions $(\delta_k)_{0\leq k\leq\rank{\Gamma}}$ into a family $(\delta'_k)_{k\in \N}$ with $\delta'_k:T_{\Gamma}^{(k)}(\{x_1\})\rightarrow S^{Q^k\times Q}$. We use the same notations for $\delta'$ as for $\delta$ and define $\delta'_k$ recursively as follows. 

\begin{enumerate}
  \item For all $g\in \Gamma^{(0)}$ let $\delta'_g(\varepsilon,q) =  \delta_g(\varepsilon,q)$ for all $q\in Q$,
  \item $\delta'_{x_1}(p,q) =  \chr{\{q\}}(p)$ for all $(p,q)\in Q\times Q$, and
  \item if $t=g(t_1,\ldots,t_k)$ for some $g\in \Gamma^{(k)}$, we define \[\delta'_t(\overbar{p}_1,\ldots,\overbar{p}_k,q) =\sum\limits_{\overbar{p}'\in Q^k}\delta_g(\overbar{p}',q) \cdot\prod_{i=1}^k \delta'_{t_i}(\overbar{p}_i, p'_i)\enspace\]for all $(\overbar{p}_1,\ldots,\overbar{p}_k,q)\in Q^{\ra(t_1)}\times\ldots\times Q^{\ra(t_k)}\times Q$.
\end{enumerate}Please note, for $t\in T_{\Gamma}$ we have $\delta'_t(\varepsilon,q)=\bhvs{\A}_q(t)$ by definition.

Now, we are well-equipped to define the \wafa $\A=(Q,\Sigma, \alpha,P_0,\tau)$. Let $P_0=\sum_{i=1}^n {\lambda(q_i)\cdot q_i}$, $\tau(q_i)=\delta'_{h(\#)}(\varepsilon,q_i)$ for all $1\leq i\leq n$, and \newline
$\alpha(q,a)=\sum_{\overbar{p} \in Q^{\ra(h(a))}} {\delta'_{h(a)}(\overbar{p},q) \cdot \prod_{i=1}^{\ra(h(a))}p_i}$ for all $q\in Q$, $a \in \Sigma$.

Next, we show $\bhvs{\B}_q(h(w))=\bhvs{\A}_q(w)$ for all $q \in Q$, $w \in \Sigma^*$. Again, by induction on $|w| \in \N$. Before we can do so, we have to prove the following auxiliary claim about $\delta'$.
\begin{clm}\label{clm:aux1} Let $\A=\wftaa$ be a \wfta, $t\in T_{\Gamma}$, $\hat{t}\in T_{\Gamma\cup \{x_1\}}$, and $t_1, \ldots,t_k \in T_{\Gamma}$ such that $t=\hat{t}\langle x_1\leftarrow(t_1,\ldots,t_k)\rangle$.
  Then \[\bhvs{\A}_q(t)= \sum_{\overbar{p}\in Q^k} \delta'_{\hat{t}}(\overbar{p},q)\cdot\prod_{i=1}^k \bhvs{\A}_{p_i}(t_i)\] for all $q\in Q$.
\end{clm}

\begin{proof}We prove this claim via induction on the depth of $\hat{t}$. The claim is clear if $k=0$. Therefore, we assume $k>0$. If $\hat{t}=x_1$, we know $t=t_1$ and get \[\begin{array}{ll}&\bhvs{\A}_q(t)\\=&\sum_{p\in Q} \chr{\{q\}}(p) \cdot \bhvs{\A}_p(t_1)\\
=&\sum_{p\in Q} \delta'_{x_1}(p,q) \cdot \bhvs{\A}_p(t_1) \enspace.\end{array}\]
Now, assume the claim holds for all $\hat{t}\in T_{\Gamma}$ with depth $\leq n$ for some $n\in \N$.
Assume $\hat{t}$ has depth $n+1$ and $\hat{t}=g(\hat{t}_1,\ldots,\hat{t}_r)$ for some $g\in \Gamma^{(r)}$ and $\hat{t}_1,\ldots,\hat{t}_r \in T_{\Gamma\cup\{x_1\}}$. Of the $k$ trees only the first $\ra(\hat{t}_1)$ will be substituted in $\hat{t}_1$, only the second $\ra(\hat{t}_2)$ will be substituted in $\hat{t}_1$, and so on. To mark this, let $k(i)=\sum_{l=1}^{i}\ra(\hat{t}_l)$, $\overbar{t}_i=(t_{k(i-1)+1},\ldots,t_{k(i)})$, and $\overbar{p}_i=(p_{k(i-1)+1},\ldots,p_{k(i)})$ for all $1\leq i\leq r$. We get:
\[\begin{array}{ll}&\bhvs{\A}_q(t)\\=&\sum\limits_{\overbar{p}'\in Q^{r}} \delta_g(\overbar{p}',q)\cdot \prod\limits_{i=1}^{r}\bhvs{\A}_{p'_i}(\hat{t}_i\langle x_1\leftarrow\overbar{t}_i\rangle)\\\overset{\mathclap{\textup{(IH)}}}{=}&\sum\limits_{\overbar{p}'\in Q^{r}} \delta_g(\overbar{p}',q)\cdot\prod\limits_{i=1}^{r}\sum\limits_{\overbar{p}_i\in Q^{\ra(\hat{t}_i)}}\delta'_{\hat{t}_i}(\overbar{p}_i,p'_i)\cdot\prod\limits_{j=r(i-1)+1}^{r(i)}\bhvs{\A}_{p_{j}}(t_j)\\=&\sum\limits_{\overbar{p}'\in Q^{r}} \delta_g(\overbar{p}',q)\cdot \sum\limits_{(\overbar{p}_1,\ldots,\overbar{p}_r)\in Q^{k}}\prod\limits_{i=1}^{r}\big(\delta'_{\hat{t}_i}(\overbar{p}_i,p'_i)\big)\cdot\prod\limits_{j=1}^{k}\bhvs{\A}_{p_{j}}(t_j)\\=& \sum\limits_{(\overbar{p}_1,\ldots,\overbar{p}_r)\in Q^{k}} \sum\limits_{\overbar{p}'\in Q^{r}}\Big( \delta_g(\overbar{p}',q)\cdot \prod\limits_{i=1}^{r}\big(\delta'_{\hat{t}_i}(\overbar{p}_i,p'_i)\big)\cdot\prod\limits_{j=1}^{k}\bhvs{\A}_{p_{j}}(t_j)\Big)\\=&\sum\limits_{(\overbar{p}_1,\ldots,\overbar{p}_r)\in Q^{k}} \sum\limits_{\overbar{p}'\in Q^{r}}\Big(\delta_g(\overbar{p}',q)\cdot \prod\limits_{i=1}^{r}\big(\delta'_{\hat{t}_i}(\overbar{p}_i,p'_i)\big)\Big)\cdot\prod\limits_{j=1}^{k}\bhvs{\A}_{p_{j}}(t_j)\\=&\sum\limits_{(\overbar{p}_1,\ldots,\overbar{p}_r)\in Q^{k}} \delta'_{\hat{t}}(\overbar{p}_1,\ldots,\overbar{p}_r,q)\cdot\prod\limits_{j=1}^{k}\bhvs{\A}_{p_{j}}(t_j)\\=&\sum\limits_{\overbar{p}\in Q^{k}} \delta'_{\hat{t}}(\overbar{p},q)\cdot\prod\limits_{j=1}^{k}\bhvs{\A}_{p_{j}}(t_j)\end{array}\] This finishes the proof of our claim.
\end{proof}
We return to our main proof. If $|w|=0$, then $w=\varepsilon$. Therefore, $\bhvs{\B}_q(w)=\tau(q)=\delta'_{h(\#)}(\varepsilon,q)\overset{*}{=} \bhvs{\A}_q(h(w))$. Here $*$ holds due to Claim \ref{clm:aux1} with $k=0$ and $\hat{t}=t=h(\#)$. Assume there exists some $l\in \N$ such that $\bhvs{\B}_q(h(w))=\bhvs{\A}_q(w)$ for all $q \in Q$, $w \in \Sigma^l$. For the induction step let $w=aw'\in \Sigma\cdot \Sigma^l$ and $k=\ra(h(a))$. Similar to the first direction, we have \[\begin{array}{ll}&\bhvs{\A}_q(aw')\\=& \sum\limits_{\overbar{p}\in Q^{k}} \delta'_{h(a)}(\overbar{p},q) \prod\limits_{i=1}^{k}\bhvs{\A}_{p_i}(w')\\ \overset{\mathclap{\textup{(IH)}}}{=}&\sum\limits_{\overbar{p}\in Q^{k}} \delta'_{h(a)}(\overbar{p},q) \prod\limits_{i=1}^{k}\bhvs{\B}_{p_i}(h(w'))\\ \overset{\mathclap{\text{C.}\ref{clm:aux1}}}{=}&\bhvs{\B}_q(h(a)\langle h(w')\rangle)\\=&\bhvs{\B}_{q}(h(w))\end{array}\] for all $q\in Q$. Via induction, this finishes our proof of the second direction.

Finally, for all $w\in \Sigma^*$ we get:
\begin{align*}
  \bhv{\B}(h(w)) &= \sum_{i=1}^n \lambda(q_i) \cdot \bhvs{\B}_{q_i}(h(w))
  = \sum_{i=1}^n \lambda(q_i) \cdot \bhvs{\A}_{q_i}(w)\\
                 &= P_0\langle\bhvs{\B'}_{q_1}(w),\ldots,\bhvs{\A}_{q_n}(w)\rangle
  = \bhv{\A}(w)
\end{align*}
\end{proof}

This leads us to our main result.

\begin{thm} \label{thr:WafaIffWfta} A weighted language $s\in \ser{S}{\Sigma^*}$ is recognized by a \wafa if and only if there exists a ranked alphabet $\Gamma$, a tree homomorphism $h: \Sigma^* \rightarrow \tlg$, and a \wfta $\A=\wftaa$ such that $s=\bhv{\A}\circ h$.
\end{thm}

\begin{proof} This is an immediate consequence of Lemma \ref{lmm:d1WafaIffWfta} and Lemma \ref{lmm:d2WafaIffWfta}.
\end{proof}

This result allows us to transfer results from \wfta to \wafa. Moreover, additional observations in the proofs show that one can give a weight preserving, bijective mapping between the runs of $\A$ and $\B$. This allows us to translate results about runs of \wfta into results about runs of \wafa. 
To see how to do this, we first revisit a well known result for non-weighted \wafa.

\begin{cor} Non-weighted alternating automata and finite automata are equally expressive.
\end{cor}

This result goes back to \cite{alternation} and deals with non-weighted automata.
It is well known that the non-weighted setting corresponds to the weighted setting, if weights are taken from the Boolean semiring.
Thus, we may still investigate the non-weighted setting with our methods.
Therefore, we can reproduce this result by using the Boolean semiring.
In doing so, we will get a better understanding of the translation of \wafa into \wfa.
In particular, we are able to identify cases, in which this translation is efficient.

Clearly, every \wfa is a \wafa, hence we only have to concern ourselves with one direction.
First, we observe that $\chr{L} \circ h=\chr{h^{-1}(L)}$ for any tree language $L\subseteq T_\Lambda$ and tree homomorphism $h:T_\Gamma\rightarrow T_\Lambda$.
Assume $\chr{L}$ is recognized by a WAFA. By Lemma \ref{lmm:d1WafaIffWfta}, we get $\chr{L}=\chr{L'}\circ h = \chr{h^{-1}(L')}$ with $h=h^r$ for some $r\in \N$ and $L'\subseteq T_{\Sigma^r_\#}$ regular. Since regular tree languages are closed under inverses of homomorphisms, we know $L = h^{-1}(L')$ is a regular $\Sigma^1_\#$ tree language.
It is also known that \wfa and \wfta are equally expressive over $\Sigma_\#^1$, thus $\chr{L}$ is recognized by a \wfa. The authors in \cite{alternation} also show that the translation of an alternating finite automaton into an equivalent deterministic finite automaton leads to a (worst case) doubly exponential blowup in states. 
By our proof we get a better understanding of where this blowup comes from.
Constructing $\B$ (from the Proof of Lemma \ref{lmm:d1WafaIffWfta}) is linear in states.
However, to construct the tree automaton $\B'$ recognizing $h^{-1}(L)$, we get an exponential blowup in states.
Next, $\B'$ viewed as a \wfa is non-deteministic (even if $\B'$ was bottom up deterministic).
Thus, another exponent is needed to determinize $\B'$.
Immediately, we see that the translation of an alaternating automaton into a non-deterministic finite automaton is only exponential.
Moreover, the first exponent is only needed, if $\B$ is not (bottom up) deterministic.
If a nice Boolean \wafa has only one non-zero final weight and for every pair of states $p,q$ we have $M_{(q,a)} \cap M_{(p,a)} = \emptyset$ for all $a \in \Sigma$, then $\B$ is bottom up deterministic.
Consequently, the translation of alternating automata with this property into a non-deterministic finite automaton is linear in states.

\section{A Nivat theorem for \wafa}
This section leads to the Nivat-like characterization of \wafa (Theorem \ref{thr:nivatWAFA}), but first we will prove that weighted languages recognized by \wafa are closed under inverses of homomorphisms (Corollary \ref{crl:ClIHWafa}), but not under homomorphisms (Lemma \ref{lmm:NClHWafa}).

Let $s_1\odot s_2$ denote the \textit{Hadamard product (pointwise product)} of two weighted languages $s_1,s_2\in \ser{S}{\Sigma^*}$.
Furthermore, a word homomorphism $h:\Gamma^*\rightarrow \Sigma^*$ is called \text{non-deleting} if and only if $h(a)\neq \varepsilon$ for all $a \in \Sigma$.
Let $r \in \ser{S}{\Gamma^*}$.
For $h:\Gamma^*\rightarrow \Sigma^*$ a non-deleting homomorphism, we define $h(r)\in \ser{S}{\Sigma^*}$ by \(h(r)(w)=\sum_{v \in h^{-1}(w)} r(v)\) for all $w \in \Sigma^*$.
This sum is always finite since $h$ is non-deleting.
For $h:\Sigma^*\rightarrow\Gamma^*$ we define $h^{-1}(r)\in \ser{S}{\Sigma^*}$ by \(h^{-1}(r)(w) = r(h(w))\) for all $w \in \Sigma^*$.
In the boolean setting, we have $h(\chr{L}) (w)=1$ if and only if there exists $v\in \Gamma^*$ with $h(v)=w$ and $v\in L$.
Thus, $h(\chr{L})=\chr{h(L)}$. Analogously, we get $h^{-1}(\chr{L}) = \chr{h^{-1}(L)}$. Hence, $h(r)$ corresponds to the application of a homomorphism, while $h^{-1}(r)$ corresponds to the application of the inverse of a homomorphism in the non-weighted setting.

The original Nivat Theorem \cite{Nivat} characterizes word-to-word transducers. A generalized version for \wfa over arbitrary semirings (Theorem 6.3 in \cite{droste2013weighted}) can be stated in the following way:

\begin{thm}[Nivat-like theorem for \wfa \cite{droste2013weighted}]\label{NfWFA} A weighted language $s \in \ser{S}{\Sigma^*}$ is recognized by a \wfa if and only if there exist an alphabet $\Gamma$, a non-deleting homomorphism $h:\Gamma^*\rightarrow \Sigma^*$, a regular language $L \subseteq\Gamma^*$, and a \wfa $\A_w$ with exactly one state such that:\[s = h(\bhv{\A_w} \odot \chr{L}) \enspace.\]
\end{thm}

Please note, $\A_w$ does not depend on any input and is called $\A_w$ since it is responsible for the application of weights. Our goal is to generalize this result to \wafa. This Nivat-like theorem is strongly connected to the closure of weighted languages recognized by \wfa under (inverses) of homomorphisms. Thus, we will investigate these properties for \wafa.

\subsection{Closure properties}
A class $K$ of $S$-weighted languages is said to be \textit{closed under homomorphisms} if $s \in \ser{S}{\Gamma^*}\cap K$  and $h:\Gamma^*\rightarrow \Sigma^*$ a non-deleting homomorphism implies $h(s) \in \ser{S}{\Sigma^*} \cap K$. Moreover, $K$ is \textit{closed under inverses of homomorphisms} if $s' \in \ser{S}{\Gamma^*} \cap K$ and $h:\Sigma^*\rightarrow \Gamma^*$ a homomorphism implies $h^{-1}(s') \in \ser{S}{\Sigma^*} \cap K$. The same notions are used for weighted tree languages.

The class of weighted languages recognized by \wfa is closed under (inverses of) homomorphisms (Lemma 6.2 in \cite{droste2013weighted}). \wafa are also closed under inverses of homomorphisms. In fact, this is an easy corollary of Lemma \ref{lmm:d1WafaIffWfta} and Lemma \ref{lmm:d2WafaIffWfta}.

\begin{cor}\label{crl:ClIHWafa} The class of weighted languages recognized by \wafa is closed under inverses of homomorphisms.
\end{cor}

\begin{proof} Let $h':\Lambda^*\rightarrow \Sigma^*$ be a homomorphism and $s\in \ser{S}{\Sigma^*}$ recognized by a \wafa.
  Due to Lemma \ref{lmm:d1WafaIffWfta}, we get $s=\bhv{\B}\circ h^r$ for the generic homomorphism $h^r: \Sigma^* \rightarrow \Sigma_\#^r$ and some weighted $\Sigma_\#^r$-\wfta $\B$.
  Since homomorphisms are closed under composition, $h^r\circ h':\Lambda^*\rightarrow \Sigma^r_\#$ is a tree homomorphism.
  Thus, due to Lemma \ref{lmm:d2WafaIffWfta}, ${h'}^{-1}(s)=(\bhv{\B}\circ h^r)\circ h'=\bhv{\B}\circ (h^r\circ h')$ is recognized by a \wafa.
\end{proof}

However, the same is not true for the closure under homomorphisms.

\begin{lem}\label{lmm:NClHWafa} The class of weighted languages recognized by \wafa is not closed under homomorphisms.
\end{lem}

\begin{proof} Let $\Sigma=\{a,b,\#\}$, $\mathbb{B}$ the Boolean semiring, and $\mathbb{B}[x]$ the semiring of polynomials in one indeterminate. Consider \[r_B: \Sigma^*\rightarrow \mathbb{B}[x]: w \mapsto \begin{cases}\sum\limits_{k=0}^{j} x^{ki} &\text{if } w=a^i\#b^j \enspace, \\ 0 &\text{otherwise.}\end{cases}\] Due to Lemma 8.3 from \cite{KOSTOLANYI20181}, we know $r_B$ is not recognized by a \wafa. Let $\Gamma=\{a,c,d,\#\}$, $h:\Gamma^*\rightarrow \Sigma^*$ the non-deleting homomorphism induced by $h(a)=a,h(\#)=\#, h(c)=h(d)=b$, and \[r_R: \Gamma^*\rightarrow \mathbb{B}[x]:w \mapsto \begin{cases} x^{ki} &\text{if } w=a^i\#c^kd^l \enspace, \\ 0 &\text{otherwise.}\end{cases}\] Then 
\[\begin{array}{llll}
   &h(r_R)(w) & = &\sum\limits_{v\in h^{-1}(w)} r_R(v)\\
  =&\begin{cases}\sum\limits_{v\in h^{-1}(w)} r_R(v) &\text{if } w=a^i\#b^j \\ 0 &\text{otherwise}\end{cases}
   &=&\begin{cases}\sum\limits_{k=0}^{j} r_R(a^i\#c^kd^{j-k}) &\text{if } w=a^i\#b^j \\ 0 &\text{otherwise}\end{cases}\\
  =&\begin{cases}\sum\limits_{k=0}^{j} x^{ki} &\text{if } w=a^i\#b^j \\ 0 &\text{otherwise}\end{cases} 
   &=&r_B(w)
\end{array}\]
for all $w \in \Sigma^*$.
Thus, $h(r_R)$ is not recognized by a \wafa. The weighted language $r_R$ is recognized by the \wafa $\A_R =(\{q_\iota,q_1,q_a,q_c,q_d\},\{a,\#,c,d\},\delta_R,q_\iota,\tau_R)$ with $\delta_R$ and $\tau_R$ defined by:
\[\begin{array}{l|c|c|c|c|c}
&q_\iota &q_1 &q_a& q_c&q_d\\
\hline
\delta_R(*,a)&q_\iota q_a &0 &q_a &0 &0\\
\hline
\delta_R(*,\#)&q_1 &0 &q_c &0 &0\\
\hline
\delta_R(*,c)&0 &q_1 &0 &x\cdot q_c &0\\
\hline
\delta_R(*,d)&0 &q_d &0 &q_d &q_d\\
\hline
\tau_R(*)&0 &1 &0 &1 &1\\
\end{array}
\]
A depiction of $\A_R$ can be seen below. This completes our proof.
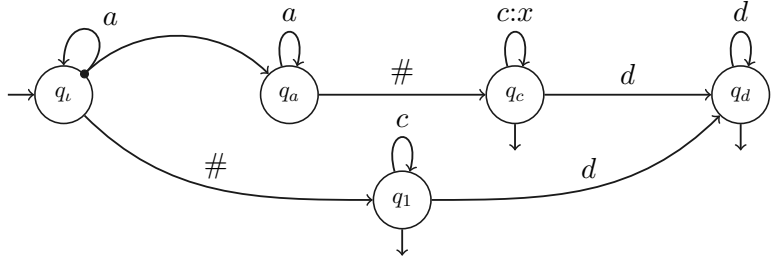
\[
% the figure had a forced position ([h!]) and wasn't referenced anywhere else. Thus,
% I'm replacing it with a math display, because now I can nicely align the automaton
% with the \qed box
\begin{tikzpicture}[every edge/.style={main edge},
          every loop/.style={main edge}, 
          every initial by arrow/.style={main edge, initial distance= 10pt},
          every accepting by arrow/.style={main edge, accepting distance= 10pt}]
\def\x{3}
\def\y{2}
            
            \node[main node,initial by arrow, initial where=left, initial text=] (qio) at (-2*\x,0*\y) {$q_\iota$};
            \begin{scope}[
                every accepting by arrow/.append style={
                  every node/.append style={
                    alias=q1accepting, % so one could also set baseline=(q1accepting.north) in the tikzpicture
                    overlay, % do not consider the (anyway invisible) text node for bounding box computation
                  }
                },
              ]
            \node[main node,
              accepting by arrow, accepting where= below, accepting text= ,
              ] (q1) at (-0.5*\x,-0.7*\y) {$q_1$};
            \end{scope}
            \node[main node] (qa) at (-1*\x,0*\y) {$q_a$};
            \node[main node, accepting by arrow, accepting where= below, accepting text= ] (qc) at (0*\x,0*\y) {$q_c$};
            \node[main node, accepting by arrow, accepting where= below, accepting text= ] (qd) at (1*\x,0*\y) {$q_d$};

            \path[main edge] (qio) edge[out=-45, in=180] node[edge label, above] {$\#$} (q1)
                  (qio) edge[out=45, in=135] node[edge label, above, xshift=-25pt] {$a$} (qa)
                  (qa) edge node[edge label, above] {$\#$} (qc)        
                  (qc) edge node[edge label, above] {$d$} (qd)
                  (qio) edge[loop, out=45, in=90, looseness=7] (qio)
                  (qa) edge [loop above] node[edge label, above] {$a$} (qa)
                  (qc) edge[loop above] node[edge label, above] {$c$:$x$} (qc)
                  (qd) edge[loop above] node[edge label, above] {$d$} (qd)
                  (q1) edge[loop above] node[edge label, above] {$c$} (q1)
                  (q1) edge[out=0, in=-135] node[edge label, above] {$d$} (qd);
            \draw[draw= grau2, fill=grau2] ([rotate around={45:(qio)}]qio.east) circle (1.5pt);
\end{tikzpicture}
\tag*{\qedhere}
\]
\end{proof}

Nonetheless, the proof of the second direction of Theorem \ref{NfWFA} relies on the closure under homomorphisms. Thus, due to Lemma \ref{lmm:NClHWafa}, a one to one translation of Theorem \ref{NfWFA} into the framework of alternating automata is prohibited. Moreover, in the proof of the first direction of Theorem \ref{NfWFA}, $L$ is defined as a language of runs of $\A$. As mentioned above, runs of \wafa are trees. Therefore, we will utilize a Nivat-like theorem for \wfta to prove the corresponding result for \wafa.

\subsection{A Nivat-like characterization of \wfta}
Nivat-like characterizations for weighted tree languages have been investigated in the past. Unranked trees were considered in \cite{DBLP:conf/birthday/DrosteG17}, while a very general result for graphs can be found in \cite{10.1007/978-3-662-48057-1_15}. Here, for the readers convenience, we want to restate a more restricted version for ranked trees. 

Let $h:\tlg\rightarrow T_\Lambda$ be a non-deleting tree homomorphism, $s\in \ser{S}{T_\Lambda}$. In analogy to words, we define $h(s) \in \ser{S}{T_\Gamma}$ by \(h(s)(t) = \sum_{t' \in h^{-1}(t)} s(t')\) for all $t \in T_\Gamma$. For linear homomorphisms the following is known: 

\begin{lemC}[{\cite[Theorem 3.8]{HWAc9}}]\label{lmm:TcluioH} The class of weighted tree languages recognized by \wfta is closed under linear homomorphisms.
\end{lemC}

Based on this, it is easy to prove the following result:

\begin{thm}[{Nivat-like theorem for \wfta \cite[Theorem 12]{DBLP:conf/birthday/DrosteG17}}]\label{thr:nivatWFTA} A weighted tree language $s \in \ser{S}{T_\Gamma}$ is recognized by a \wfta if and only if there exist a ranked alphabet $\Lambda$, a linear tree homomorphism $h:T_\Lambda\rightarrow T_\Gamma$, a regular tree language $L \subseteq T_\Lambda$, and a \wfta $\A_w$ with exactly one state such that:\[s = h(\bhv{\A_w} \odot \chr{L}) \enspace.\]
\end{thm}

\begin{proof}  A proof can be found in \cite{DBLP:conf/birthday/DrosteG17}. There, the $\Rightarrow$-direction is proved based on a \wfta $\A$ recognizing $s$. The components $\Lambda$, $h$, $L$, and $\A_w$ are chosen to be the set of transitions in $\A$, the mapping assigning each transition $(\overbar{p}, \gamma, q)$ to the corresponding letter $\gamma \in \Gamma$, the tree language of runs in $\A$, and an automaton adding weights to the transitions, respectively. This yields the desired equation. Here, for the convenience of the reader, we want to give a precise construction, as well as a proof of correctness for the second direction of the proof of Theorem \ref{thr:nivatWFTA}.

Let $s$ be recognized by a \wfta $\A=(Q,\Gamma,\alpha,\lambda)$. First, we define $\Lambda$. Let $P=Q\cup Q_{\text{fin}}$ where $Q_{\text{fin}}$ is a disjoint copy of $Q$ including an element $q_{\text{fin}}$ for each $q\in Q$. Now, $\Lambda= \bigcup_{i=0}^{\rank(\Gamma)} Q^i\times \Gamma^{(i)} \times P$. For reasons of readability, we sometimes abbreviate a letter $\big((q_1,\ldots,q_{\rank(g)}),g,p\big)\in \Lambda$ by $[\overbar{q},g,p]$. We define $\rank([\overbar{q},g,p])=\rank(g)$. Let $h$ be the tree homomorphism induced by $h\big([\overbar{q},g,p](x_1,\ldots,x_k)\big)=g(x_1,\ldots,x_k)$ for all $g\in \Gamma^{(k)}$. Clearly, $h$ is linear and non-deleting.

Next, we consider the \wfta $\B$ with weights taken from the Boolean semiring $\mathbb{B}$. Let $\B=(P, \Lambda, \beta,\chr{Q_\text{fin}})$ with $\beta$ defined by \[\beta_{[\overbar{q},g,p]}(\overbar{q}',p')=\begin{cases}1 &\text{if } [\overbar{q}',g,p']=[\overbar{q},g,p]\enspace,\\ 0 &\text{otherwise}\end{cases}\] for all $[\overbar{q},g,p]\in \Lambda$, $(\overbar{q}',p')\in Q^{k+1}$. Clearly, $\bhv{\B}=\chr{L}$, for some tree language $L\subseteq T_\Lambda$. It is well known, that the support of a recognizable $\mathbb{B}$-weighted weighted tree language is a recognizable tree language (cf. Lemma 3.11 \& Theorem 3.12 in \cite{HWAc9}). Hence, $L$ is a recognizable tree language.

Finally, we define the \wfta $\A_w=(\{q_w\}, \Lambda, \omega, \chr{\{q_w\}})$ with \[\omega_{[\overbar{q},g,p]}\big((\underbrace{q_w,\ldots, q_w}_{\rank(g) \text{ times}}),q_w\big)=\begin{cases}\alpha_g(\overbar{q},q')&\text{if } p=q'\in Q \\ \alpha_g(\overbar{q},q')\cdot\lambda(q') &\text{if } p=q'_{\text{fin}} \in Q_{\text{fin}}\end{cases}\] for all $[\overbar{q},g,p]\in \Lambda$. Please note that $\bhv{\A_w}=\bhvs{\A_w}_{q_w}$. Thus, we will not distinguish between those two semantics.

To prove that \[\bhv{\A} = h^{-1}\circ(\bhv{\A_w} \odot \chr{L})\] holds for this construction, we will prove the following claim first.

\pagebreak[2]
\begin{clm}\label{clm:aux2}
  For all $q\in Q$ it holds $\bhvs{\A}_q(t) = \big(h^{-1}\circ(\bhv{\A_w} \odot \bhvs{\B}_q)\big)(t)$ for all
  $t\in T_\Gamma$.
\end{clm}
\begin{proof} The proof is by induction on the depth of $t$. If $t=g$ for some $g\in \Gamma^{(0)}$, we get \[\begin{array}{ll} &\bhvs{\A}_q(t)\\ = &\alpha_g(\varepsilon,q) + 0\\ = &\sum\limits_{p\in Q}\big(\alpha_g(\varepsilon,p)\cdot \chr{\{q\}}(p)\big) + \sum\limits_{p_{\text{fin}}\in Q_{\text{fin}}}\big(\alpha_g(\varepsilon,p)\cdot\lambda(p)\cdot \chr{\{q\}}(p_{\text{fin}})\big)\\=&\sum\limits_{[\varepsilon,g,p]\in \Lambda^{(0)}}\bhv{\A}_w([\varepsilon,g,p])\cdot \bhvs{\B}_q([\varepsilon,g,p])\\=&\sum\limits_{t'\in h^{-1}(g)}\big(\bhv{\A}_w\odot \bhvs{\B}_q\big)(t')\\=&\big(h^{-1}\circ(\bhv{\A_w} \odot \bhvs{\B}_q)\big)(g)\end{array}\] for all $q\in Q$.

Assume the claim holds for all $t \in T_\Gamma$ of depth lower or equal to $n$ for some $n\in \N$.

We consider some $t \in T_\Gamma$ of depth $n+1$. There exist $g\in \Gamma^{(k)}$ and $t_1,\ldots,t_k \in T_\Gamma$ such that $t=g(t_1,\ldots,t_k)$. Clearly, $t_i$ has a depth lower or equal to $n$ for all $1\leq i\leq k$. Let us denote the tuple $(t_1,\ldots,t_k)$ by $\overbar{t}$ and $h^{-1}(t_1) \times \ldots \times h^{-1}(t_k)$ by $h^{-1}(\overbar{t})$. First, we observe \[t'\in h^{-1}(t) \Leftrightarrow t'=[\overbar{p},g,q](\overbar{t'})\] for some $q\in P$, $\overbar{p}\in Q^k$, and $\overbar{t'} \in h^{-1}(\overbar{t})$. Moreover, we have \[\bhvs{\B}_q\big([\overbar{p},g,q'](\overbar{t'})\big)=1 \Leftrightarrow q=q' \wedge \forall 1\leq i \leq k.\big(\bhvs{\B}_{p_i}(t_i')=1\big)\enspace .\] Therefore, \begin{equation}\label{qtn:sumOStuff}\sum_{t'\in h^{-1}(t)} r\odot \bhvs{\B}_q(t') = \sum_{\overbar{p}\in Q^k} \sum_{\overbar{t'}\in h^{-1}(\overbar{t})} r\odot \bhvs{B}_q\big([\overbar{p},g,q](\overbar{t'})\big)\end{equation} holds for all $r \in \ser{S}{T_\Lambda}$.

By this, we get \[\begin{array}{ll} &\bhvs{\A}_q(t)\\ = &\sum\limits_{\overbar{p} \in Q^k} \alpha_g(\overbar{p},q) \prod\limits_{i=1}^k \bhvs{\A}_{p_i}(t_i)\\ \overset{\mathclap{\textup{(IH)}}}{=} &\sum\limits_{\overbar{p} \in Q^k} \alpha_g(\overbar{p},q) \prod\limits_{i=1}^k \big(h^{-1}\circ(\bhv{\A_w} \odot \bhvs{\B}_{p_i})\big)(t_i)\\= &\sum\limits_{\overbar{p} \in Q^k} \alpha_g(\overbar{p},q) \prod\limits_{i=1}^k \big(\sum\limits_{t_i' \in h^{-1}(t_i)}\bhv{\A_w} \odot \bhvs{\B}_{p_i}(t_i')\big)\\ = &\sum\limits_{\overbar{p} \in Q^k} \alpha_g(\overbar{p},q) \sum\limits_{\overbar{t'} \in h^{-1}(\overbar{t})} \prod\limits_{i=1}^k \bhv{\A_w} \odot \bhvs{\B}_{p_i}(t_i')\\= &\sum\limits_{\overbar{p} \in Q^k} \sum\limits_{\overbar{t'} \in h^{-1}(\overbar{t})} \big(\alpha_g(\overbar{p},q) \cdot \prod\limits_{i=1}^k \bhv{\A_w}(t_i')\big) \cdot \big(1\cdot \prod\limits_{i=1}^k\bhvs{\B}_{p_i}(t_i')\big)\\= &\sum\limits_{\overbar{p} \in Q^k} \sum\limits_{\overbar{t'} \in h^{-1}(\overbar{t})} \bhv{\A_w}\big([\overbar{p},g,q](\overbar{t'})\big) \cdot \bhvs{B}_q\big([\overbar{p},g,q](\overbar{t'})\big)\\ \overset{\mathclap{\ref{qtn:sumOStuff}}}{=} & \sum\limits_{t'\in h^{-1}(t)}\bhv{\A_w} \odot \bhvs{\B}_q(t')\\=& h^{-1}\circ\big(\bhv{\A_w}\odot \bhvs{\B}_q\big)(t)\end{array}\]
This finishes the proof of the claim.
\end{proof}
We return to our main proof, it remains to show that we are able to apply the final weights. We consider some $t\in T_\Gamma$. There exists $g\in\Gamma^{(k)}$ and $\overbar{t} \in T_\Gamma^k$ (we use the same notation as above) with $t=g(\overbar{t})$. First, we make a similar observation as in the proof of our claim. Namely, \[\bhv{\B}\big([\overbar{q},g,p](\overbar{t'})\big)=1 \Leftrightarrow p\in Q_{\text{fin}} \wedge \forall{1\leq i\leq k}.\big(\bhvs{\B}_{q_i}(t_i')=1\big)\] for all $\overbar{t'} \in h^{-1}(\overbar{t})$. Thus,\begin{equation}\label{qtn:sumOStuff2}\sum_{t'\in h^{-1}(t)} r\odot\bhv{\B}(t') = \sum_{p\in Q}\sum_{\overbar{q}\in Q^k}\sum_{\overbar{t'}\in h^{-1}(\overbar{t})} r\odot \bhv{\B}\big([\overbar{q},g,p_{\text{fin}}](\overbar{t'})\big)\end{equation} for all $r\in \ser{S}{T_\Lambda}$. Finally, analogously to the proof of Claim \ref{clm:aux2}, we are able to deduce \[\begin{array}{ll}&\bhv{\A}(t)\\=&\sum\limits_{p\in Q} \lambda(p)\cdot{\bhvs{\A}_p(t)}\\=&\sum\limits_{p\in Q}\sum\limits_{\overbar{q}\in Q^k}\lambda(p)\cdot \alpha_g(\overbar{p},q)\cdot\prod\limits_{i=1}^k\bhvs{\A}_{q_i}(t_i)\\ \overset{\mathclap{\text{C. \ref{clm:aux2}}}}{=} &\sum\limits_{p\in Q}\sum\limits_{\overbar{q}\in Q^k}\lambda(p)\cdot \alpha_g(\overbar{p},q)\cdot\prod\limits_{i=1}^k h^{-1}\big(\bhv{\A_w}\odot\bhvs{\B}_{q_i}\big)(t_i)\\= &\sum\limits_{p\in Q}\sum\limits_{\overbar{q}\in Q^k}\lambda(p)\cdot \alpha_g(\overbar{p},q)\cdot\prod\limits_{i=1}^k \sum\limits_{t_i'\in h^{-1}(t_i)}\bhv{\A_w}\odot\bhvs{\B}_{q_i}(t_i')\\= &\sum\limits_{p\in Q}\sum\limits_{\overbar{q}\in Q^k}\sum\limits_{\overbar{t'}\in h^{-1}(\overbar{t})}\big(\lambda(p)\cdot \alpha_g(\overbar{p},q)\cdot\prod\limits_{i=1}^k \bhv{\A_w}(t_i')\big)\cdot \big(1 \cdot\prod\limits_{i=1}^k\bhvs{\B}_{q_i}(t_i')\big)\\= &\sum\limits_{p\in Q}\sum\limits_{\overbar{q}\in Q^k}\sum\limits_{\overbar{t'}\in h^{-1}(\overbar{t})}\big(\omega_{[\overbar{q},g,p_{\text{fin}}]}(q_w,(q_w,\ldots,q_w))\cdot\prod\limits_{i=1}^k \bhv{\A_w}(t_i')\big)\cdot \bhv{\B}\big([\overbar{q},g,p_{\text{fin}}](\overbar{t'})\big)\\=&\sum\limits_{p\in Q}\sum\limits_{\overbar{q}\in Q^k}\sum\limits_{\overbar{t'}\in h^{-1}(\overbar{t})}\bhv{\A_w}\big([\overbar{q},g,p_{\text{fin}}](\overbar{t'})\big)\cdot \bhv{\B}\big([\overbar{q},g,p_{\text{fin}}](\overbar{t'})\big)\\=&\sum\limits_{p\in Q}\sum\limits_{\overbar{q}\in Q^k}\sum\limits_{\overbar{t'}\in h^{-1}(\overbar{t})}\bhv{\A_w}\odot \bhv{\B}\big([\overbar{q},g,p_{\text{fin}}](\overbar{t'})\big)\\\overset{\mathclap{\ref{qtn:sumOStuff2}}}{=}&\sum\limits_{t'\in h^{-1}(t)}\bhv{\A_w}\odot \bhv{\B}(t')\\=&h^{-1}\big(\bhv{\A_w}\odot\bhv{\B}\big)(t)=h^{-1}\big(\bhv{\A_w}\odot\chr{L}\big)(t) \enspace.\end{array}\] Since $t$ was arbitrary, this completes our proof.
\end{proof}

Based on this result and Theorem \ref{thr:WafaIffWfta} a characterization of \wafa via a Nivat-like Theorem is immediate.

\begin{thm}[Nivat-like theorem for \wafa]\label{thr:nivatWAFA} A weighted language $s \in \ser{S}{\Sigma}$ is recognized by a \wafa if and only if there exist a rank $r\in \N$, a ranked alphabet $\Lambda$, a linear tree homomorphism $h:T_\Lambda\rightarrow T_{\Sigma^r_\#}$, a regular tree language $L \subseteq T_\Lambda$, and a \wfta $\A_w$ with exactly one state such that for all $w \in \Sigma^*$, it holds:\[s(w) = h(\bhv{\A_w} \odot \chr{L})(\tw) \enspace.\]
\end{thm}
\begin{proof}
$\Rightarrow$: Let $\A$ be a nice, equalized \wafa such that $\bhv{\A}=s$. Due to Lemma \ref{lmm:d1WafaIffWfta}, $r\in\N$ and a \wfta $\B$ exist such that $s(w)=(\bhv{\B}\circ h^r)(w)=\bhv{\B}(t_w^r)$. Applying Theorem \ref{thr:nivatWFTA} to $\bhv{\B}$ gives us the desired result.

$\Leftarrow$: By Theorem \ref{thr:nivatWFTA}, there exists a \wfta $\B$ such that $\bhv{\B} = h(\bhv{\A_w}\odot \chr{L})$. Let $h^r:\Sigma^*\rightarrow T_{\Sigma_\#^r}$ be the generic homomorphism. In consequence of Theorem \ref{thr:WafaIffWfta}, a \wafa $\A$ exists such that $\bhv{\A}(w)= \bhv{\B}(h^r(w))=h(\bhv{\A_w}\odot\chr{L})(h^r(w))=h(\bhv{\A_w}\odot\chr{L})(t_w^r)$ for all $w\in \Sigma$. This finishes our proof.
\end{proof}

\section{A logical characterization of \wafa}
Based on Theorem \ref{thr:WafaIffWfta} we are able to give a logical characterization of \wafa (Theorem \ref{thr:WafaIffsrMSO}). For this purpose, we will use the logical characterization by weighted MSO logic for trees which was introduced in \cite{DROSTE2006228}. 

Weighted MSO logic over trees is an extension of MSO logic over trees. It allows for the use of usual MSO formulas, but also incorporates quantitative aspects such as semiring elements and operations, as well as weighted quantifiers. In the end, every weighted MSO formula defines a weighted tree language. More precisely, let $\Gamma$ be a ranked alphabet, each weighted MSO formula $\varphi \in \textup{MSO}(\Gamma,S)$ defines a weighted tree language $\bhv{\varphi}:T_\Gamma \rightarrow S$. Weighted MSO logic is strictly more expressive than \wfta. Nevertheless, it is possible to restrict the syntax of weighted MSO in such a way that it characterizes weighted tree languages recognized by \wfta. This fragment is called weighted syntactically restricted MSO (srMSO). Since it is not needed to understand the following proofs, we have omitted the formal definition of srMSO. We will use syntax and semantics of weighted srMSO without any changes and refer the interested reader to \cite{HWAc9} or \cite{droste2011weighted}. Our characterization of \wafa will be fully based on the following characterization theorem for \wfta:

\begin{thmC}[{\cite[Theorem 3.49 (A)]{HWAc9}}]\label{thr:wftaIffMSO} A weighted tree language $s\in \ser{S}{T_\Gamma}$ is recognized by a \wfta if and only if there exist $\varphi \in \textup{srMSO}(\Gamma,S)$ such that $s=\bhv{\varphi}$.
\end{thmC}

However, we still have to handle the homomorphism used in Theorem \ref{thr:WafaIffWfta}. This will be done by choosing an appropriate way of representing words as relational structures.

By definition $\bhv{\varphi}\in \ser{S}{T_\Gamma}$ for all  $\varphi \in \textup{srMSO}(\Gamma,S)$. However, we want to use weighted srMSO on trees to define weighted languages on words. To this end, we define $\bhv{\varphi}_\Sigma\in \ser{S}{T_\Gamma}$ by $\bhv{\varphi}_\Sigma(w)=\bhv{\varphi}({t_w^{\rank(\Gamma)}})$ for all $\varphi \in \textup{srMSO}(\Gamma,S)$, $w\in \Sigma^*$. Since $\textup{srMSO}(\Gamma,S)\subseteq \textup{srMSO}(\Gamma\cup \Sigma_\#^{\rank(\Gamma)},S)$, we can assume without loss of generality $\varphi\in \textup{srMSO}(\Gamma\cup \Sigma_\#^{\rank(\Gamma)},S)$. Hence, $\bhv{\varphi}_\Sigma$ is well defined for all $\Sigma$. It is easy to see that $\bhv{\varphi}_\Sigma= \bhv{\varphi}\circ h$ where $h:\Sigma \rightarrow \Sigma_\#^{\rank(\Gamma)} \cup \Gamma$  is the generic homomorphism.

\begin{thm}\label{thr:WafaIffsrMSO} A weighted language $s\in \ser{S}{\Sigma}$ is recognized by a \wafa if and only if there exist a ranked alphabet $\Gamma$ and $\varphi \in \textup{srMSO}(\Gamma,S)$ such that $s=\bhv{\varphi}_\Sigma$.
\end{thm}

\begin{proof} $\Rightarrow$: Assume $s\in \ser{S}{\Sigma}$ is recognized by a \wafa. By \ref{lmm:d1WafaIffWfta}, there exists $r\in \mathbb{N}$ and a \wfta $\mathcal{B}$ such that $s=\bhv{\mathcal{B}}\circ h^r$. By Theorem \ref{thr:wftaIffMSO}, $\varphi \in \textup{srMSO}(\Sigma_\#^r,S)$ exists such that $\bhv{\mathcal{B}}=\bhv{\varphi}$. Thus $s=\bhv{\mathcal{B}}\circ h^r=\bhv{\varphi}\circ h^r=\bhv{\varphi}_\Sigma$.

$\Leftarrow$: If $s=\bhv{\varphi}_\Sigma$ for some $\varphi\in \textup{srMSO}(\Gamma\cup \Sigma_\#^{\rank(\Gamma)},S)$, we get $s = \bhv{\varphi}\circ h^{\rank(\Gamma)}$. By Theorem \ref{thr:wftaIffMSO}, a \wfta $\mathcal{B}$ exists such that $\bhv{\varphi}=\bhv{\mathcal{B}}$. Therefore, $s=\bhv{\mathcal{B}}\circ h^{\rank(\Gamma)}$. Since $\mathcal{B}$ is a \wfta and $h^{\rank(\Gamma)}$ a homomorphism, a \wafa $\mathcal{A}$ with $s=\bhv{A}$ exists by Lemma \ref{lmm:d2WafaIffWfta}.
\end{proof}

While mirroring the branching behavior of \wafa in the logic gives a natural characterization of \wafa, the question arises how \wafa relate to the weighted MSO logic for words. It is well known that \wfta are not capable of characterizing the entirety of weighted MSO logic for words, simply because MSO logic can define series which grow doubly exponential in the size of the input. While \wafa have this ability (cf. Example \ref{xmp:serp}), they still are incapable of capturing the entirety of weighted MSO. It can be shown that the series $r_B$ from the proof of Lemma \ref{lmm:NClHWafa} can be defined in weighted MSO logic for words, while it is not recognized by \wafa. If it is possible to characterize \wafa by a fragment of weighted MSO logic for words remains open.

\section{Closure of \wfta under inverses of homomorphisms}
It is well known (cf. Theorem 1.4.4 in \cite{tata2022}) that regular tree languages are closed under inverses of homomorphisms. Sadly, this is not true in the weighted case, at least not for arbitrary semirings. This raises the question if it is possible to give a precise description of the class of semirings for which recognizable weighted tree languages are closed under inverses of homomorphisms. This question will be answered by Theorem \ref{thr:clOTAUIHom}.

Theorem \ref{thr:nivatWAFA} and Theorem \ref{thr:wftaIffMSO} used Theorem \ref{thr:WafaIffWfta} to apply known results for \wfta to \wafa. Vice versa, we can use Theorem \ref{thr:WafaIffWfta} and Theorem \ref{thr:WAFAiffWFAiffLocFin} to characterize the semirings for which the class of recognizable $S$-weighted tree languages is closed under inverses of homomorphisms.

\begin{thm}\label{thr:clOTAUIHom} The class of $S$-weighted tree languages recognized by $\wfta$ is closed under inverses of homomorphisms if and only if $S$ is locally finite.
\end{thm}

To prove this result we will use the notion of \textit{recognizable step functions}: A weighted tree language $r\in \ser{S}{T_\Gamma}$ is a recognizable step function if there exist recognizable tree languages $L_1,\ldots, L_k$ and values $l_1,\ldots, l_k\in S$ such that $r=\sum_{i=1}^k l_i\cdot\chr{L_i}$. Due to \cite{DROSTE2006228}, we know the following about recognizable step functions:

\begin{lemC}[{\cite[Lemma 3.1]{DROSTE2006228}}]\label{lmm:auxfD1} If $r\in \ser{S}{T_\Gamma}$ is a recognizable step function, then a partition $L_1,\ldots, L_k$ of $T_\Gamma$ exits such that $r=\sum_{i=1}^k l_i\cdot\chr{L_i}$ for some $l_1,\ldots,l_k \in S$.
\end{lemC}

Note, this lemma is not redundant since the definition of recognizable step functions does not demand that the recognizable tree languages are pairwise disjoint. Due to Lemma \ref{lmm:auxfD1}, we know that a weighted tree language is a recognizable step function if and only if it has a finite image and each preimage is a recognizable tree language. The next Lemma characterizes recognizable weighted tree languages over locally finite semirings.

\begin{lemC}[{\cite[Lemma 3.3 \& Lemma 6.1]{DROSTE2006228}}]\label{lmm:auxfD2} Let $S$ be locally finite. A weighted tree language $r\in \ser{S}{T_\Gamma}$ is recognizable if and only if $r$ is a recognizable step function.
\end{lemC}

Finally, we can proceed with the proof of Theorem \ref{thr:clOTAUIHom}.

\begin{proof}[Proof of Theorem \ref{thr:clOTAUIHom}]$\Rightarrow$: Assume the class of $S$-weighted tree languages recognized by $\wfta$ is closed under inverses of homomorphisms. 

\begin{clm}\label{clm:hClPauxfD2} The class of $S$-weighted $\Sigma$ languages recognizable by \wafa and the class of $S$-weighted $\Sigma$ languages recognizable by \wfa are equal.
\end{clm}

Clearly every \wfa is a \wafa. Thus, for the proof of the claim, it remains to show that every weighted language which is recognized by a \wafa is recognized by a \wfa.
For this purpose, assume $s\in \ser{S}{\Sigma^*}$ is recognized by a \wafa $\A$.
Due to Theorem \ref{thr:WafaIffWfta}, there exists a \wfta $\B$ and a homomorphism $h: \Sigma^* \rightarrow T_\Gamma$ such that $\bhv{\A}=\bhv{\B}\circ h$.
However, by our assumption there exists a \wfta $\mathcal{C}$ over $\Sigma_\#^1$ such that $\bhv{\mathcal{C}}=\bhv{\B}\circ h$ and hence a \wfa $\mathcal{C}'$ over $\Sigma$ such that $\bhv{\A}=\bhv{\B}\circ h=\bhv{\mathcal{C}}=\bhv{\mathcal{C}'}$.
Thereby, $s$ is recognized by a \wfa. This proves our claim.

By Claim \ref{clm:hClPauxfD2} and Theorem \ref{thr:WAFAiffWFAiffLocFin} it follows that $S$ is locally finite.
This finishes the proof of the first direction.

$\Leftarrow$: Assume $S$ is locally finite. Let $r\in \ser{S}{T_\Gamma}$ be recognizable and $h:T_\Lambda \rightarrow T_\Gamma$ a homomorphism. Due to Lemma \ref{lmm:auxfD1} and Lemma \ref{lmm:auxfD2}, we have $r=\sum_{i=1}^k l_i\cdot \chr{L_i}$ for some partition $L_1,\ldots,L_k$ of $T_\Gamma$ and values $l_1,\ldots,l_k \in S$. We claim $r\circ h = \sum_{i=1}^k l_i\cdot \chr{h^{-1}(L_i)}$. To prove this, consider some arbitrary $t\in T_\Lambda$. Since the $L_i$ form a partition of $T_\Gamma$ there exists a unique $j\in\{1,\ldots, k\}$ such that $h(t) \in L_i$. Therefore, we have\[\begin{array}{cl}&(r\circ h)(t)\\=&\sum_{i=1}^k l_i\cdot\chr{L_i}(h(t))=l_j\\=&l_j\cdot \chr{h^{-1}(L_j)}(t)\\\overset{l_j \text{ unique}}{=}&\sum_{i=1}^k l_i\cdot \chr{h^{-1}(L_i)} (t)\enspace.  \end{array}\] Since recognizable tree languages are closed under inverses of homomorphisms, we know that $h^{-1}(L_1), \ldots, h^{-1}(L_k) \subseteq T_{\Lambda}$ are recognizable. Thus, $r\circ h$ is a recognizable step function. Again, by Lemma \ref{lmm:auxfD2}, we get $r\circ h$ is recognizable. This completes our proof.
\end{proof}

\section{\wafa and polynomial automata}
We will use known results for polynomial automata, to prove the decidability of the Zeroness Problem for \wafa if weights are taken from the rationals (Lemma \ref{crl:DecZerEqWafa}).

Informally a polynomial automaton is a set of registers which get updated by polynomial funcions according to some input.
Historically this principle was studied from many perspectives.
For example: \cite{ppr:SeqOfLev123k} studies it in the form of ``polynomial recurrent relations'', or the ``cost register automata'' from \cite{ppr:RegFuncACoRegAut} which allow for a very broad class of updates and weight structures.
We will follow the terminology and definitions of \cite{8005101} where ``polynomial automata'' were considered as a generalization of both vector addition systems and weighted automata. Polynomial automata are quite similar to \wafa, the authors of \cite{8005101} even prove that the characteristic function of the reversal of each language recognized by a non-weighted alternating automaton is recognized by a polynomial automaton of the same size. We want to strengthen this connection. In \cite{8005101} polynomial automata are defined over the rational numbers. However, it is easy to give a more general definition of arbitrary commutative semirings.

A \textit{polynomial automaton} (\pola) is a $5$-tuple $\A=(n,\Sigma, \alpha, p, \gamma)$, where $n\in \N$ is the number of states, $\Sigma$ is an alphabet, $\alpha \in S^n$ is an initial weight vector, $p:\Sigma\rightarrow {\pols{X_n}}^n$ the transition function, and $\gamma\in \pols{X_n}$ an output polynomial. We denote the $i$-th entry of $p(a)$ by $p_i(a)$.

Let $\A=(n,\Sigma, \alpha, p, \gamma)$ be a \pola. Its \textit{state behavior} $\bhvs{\A}:\{1,\ldots, n\}\times \Sigma ^*\rightarrow S$ is the mapping defined by \[\bhvs{\A}(i,w) = \begin{cases}
  \alpha_i &\text{ if } w=\varepsilon,\\
  p_i\big\langle\bhvs{\A}(1,v),\ldots,\bhvs{\A}(n,v)\big\rangle &\text{ if } w=va \text{ for } a\in \Sigma.
\end{cases}\]
Usually we will denote $\bhvs{\A}(i,w)$ by $\bhvs{\A}_i(w)$. Now, the \textit{behavior} of $\A$ is the weighted language $\bhv{\A}:\Sigma^*\rightarrow S$ defined by \[\bhv{\A}(w) = \gamma\big(\bhvs{\A}_1(w),\ldots,\bhvs{\A}_n(w)\big).\]

It is easy to check that this definition is a reformulation of the definition found in \cite{8005101}.

Let the reversal of a weighted language $s\in \ser{S}{\Sigma}$ be defined by $s^R(w)=s(w^R)$ for all $w=w_1\ldots w_n$ with $w_1, \ldots, w_n \in \Sigma$, where $w^R=w_n\ldots w_1$. Comparing the definition of state behavior for \wafa and \pola already yields the following lemma:

\begin{lem}\label{lmm:WafaIffRPola} A weighted language $s\in \ser{S}{\Sigma}$ is recognized by a \wafa if and only if $s^R$ is recognized by a \pola.
\end{lem}

\begin{proof} Assume $s$ is recognized by $\A=\wafaa$. 

  Let $\B=(|Q|,\Sigma, \big(\tau(q_1),\ldots,\tau(q_n)\big), p, P_0)$ be a \pola with $p_i(a)=\delta(q_i,a)$ for all $1\leq i\leq |Q|, a\in \Sigma$. Then, a straightforward induction on $|w|$ shows $\bhv{A}(w)=\bhv{\B}(w^R)$ for all $w\in \Sigma$.
The second direction is proven analogously to the first one.
\end{proof}

Let $\A,\A'$ be two \wafa. We observe $\bhv{\A}(w) = 0$ for all $w\in \Sigma^*$ if and only if $\bhv{\A}(w^R)=0$ for all $w \in \Sigma^*$. Moreover, we have $\bhv{\A}(w) = \bhv{\A'}(w)$ for all $w\in \Sigma^*$ if and only if $\bhv{\A}(w^R)=\bhv{\A'}(w^R)$ for all $w \in \Sigma^*$. This allows us to derive the following corollary from Lemma \ref{lmm:WafaIffRPola}:

\begin{cor}\label{crl:DecZerEqWafa} The Zeroness Problem and the Equivalence Problem for \wafa with weights taken from the rationals are in the complexity class ACKERMANN and hard for the complexity class ACKERMANN.
\end{cor}

\begin{proof} Using the respective results for polynomial automata (Theorem 1, Theorem 4, and Corollary 1 in \cite{8005101}) together with Lemma \ref{lmm:WafaIffRPola} yields this result.
\end{proof}

In general, weighted languages recognized by \wafa are not closed under reversal (Theorem 8.4 in \cite{KOSTOLANYI20181}).
Thus, the class of weighted languages recognized by \wafa and the class of languages recognized by \pola differ.
Moreover, series recognized by polynomial automata are (in general) not close under reversal.
However, if weights are taken from a commutative semiring, series recognized by \wfa are closed under reversal, this would allow for a direct translation of Theorem \ref{thr:WAFAiffWFAiffLocFin} into the setting of \pola.

\section{Conclusion}
We were able to connect WAFA to a variety of formalisms, giving a better understanding of their expressive power and characterizing the class of quantitative languages recognized by WAFA. From here, there are various routes to take. It could be of great practical use to find a logical characterization of WAFA via a linear formalism such as a weighted linear logic, or weighted rational expressions tailored to the expressive power of WAFA, or a fitting fragment of weighted MSO logic for words.

Similar to the work in \cite{8005101}, one could investigate subclasses of WAFA allowing for more efficient decision procedures. An interesting candidate is strictly alternating \wafa who have a bounded number of alternations within any run. While these use the ability to recognize very fast growing series, they still add in expressive power due to the ability to multiply over subruns.

A different direction would aim to use the expressive power added by alternation to achieve an automata model which is as expressive weighted MSO for words. While \wafa can not fulfill this role a generalization of \wafa might be. After all, the use of product and sum within runs is very similar, to the use of product- and sum-quantifiers in weighted MSO for words.

Alternatively, one could approach the concept of alternation in weighted automata dealing with more complex structures than words, such as weighted alternating tree automata. And of course, having the universal interpretation of nondeterminism in mind, one may take several of these routes at once!
\nocite{*}
\bibliographystyle{alphaurl} %alphaurl
\bibliography{myBib}
\end{document}